\documentclass[twocolumn,final]{svjour3}         % twocolumn
\smartqed  % flush right qed marks, e.g. at end of proof

\usepackage{xspace}
\usepackage{graphicx}
\usepackage{amsmath}
\usepackage{fncylab}
\usepackage[usenames]{color}
\usepackage{hyphenat}
\usepackage{multirow}

\newtheorem{observation}[theorem]{Observation}

\definecolor{darkgray}{rgb}{0.25,0.25,0.25}
\definecolor{Darkgray}{rgb}{0.15,0.15,0.15}

\newcommand{\Alabel}[1]{\labelformat{myalgorithm}{Alg.~{\arabic{myalgorithm}}}\refstepcounter{myalgorithm}{\mbox{}\hfill{\small{}\color{darkgray}Alg.\,\small\arabic{myalgorithm}}}{\label{#1}}}

\newcommand{\algfont}{\rm}
\newcommand{\tab}{\hspace*{0.25in}}

\newcounter{myline}
\newcounter{myalgorithm}

\newenvironment{alg}{
  \medskip
  \par
  \algfont
    \noindent
      ~\begin{tabular}{|l|}\hline
      \begin{minipage}{0.98\textwidth}\raggedright
        \begin{list}{\arabic{myline}.}{
            \usecounter{myline}
            \setlength{\listparindent}{0in}
            \setlength{\topsep}{0in}
            \setlength{\itemsep}{0in}
            \setlength{\parsep}{0in}
            \setlength{\rightmargin}{0in}
            \setlength{\itemindent}{0in}
            \setlength{\labelsep}{0.035in}
            \setlength{\leftmargin}{0.215in}
          }
        \vspace*{0.06in}
          }{
        \vspace*{0.06in}
        \end{list}
      \end{minipage}\\\hline
    \end{tabular}
    \medskip
    \par
    \noindent
}
\newenvironment{alg1}{
  \medskip
  \par
  \algfont
    \noindent
      ~\begin{tabular}{|l|}\hline
      \begin{minipage}{0.47\textwidth}\raggedright
        \begin{list}{\arabic{myline}.}{
            \usecounter{myline}
            \setlength{\listparindent}{0in}
            \setlength{\topsep}{0in}
            \setlength{\itemsep}{0in}
            \setlength{\parsep}{0in}
            \setlength{\rightmargin}{0in}
            \setlength{\itemindent}{0in}
            \setlength{\labelsep}{0.03in}
            \setlength{\leftmargin}{0.05in}
          }
        \vspace*{0.06in}
          }{
        \vspace*{0.06in}
        \end{list}
      \end{minipage}\\\hline
    \end{tabular}
    \medskip
    \par
    \noindent
}
\newcommand{\A}{\item}
\newcommand{\Anonum}{\item[]}
\newcommand{\Along}[1]{\item\parbox[t]{0.95\linewidth}{ #1}}
\newcommand{\Ahead}[1]{\item[]\hspace*{-\leftmargin}{#1}}
\newcommand{\Ain}[1]{\item[]\hspace*{-\leftmargin}{~\textrm{\small{input:} #1}}}
\newcommand{\Aout}[1]{\item[]\hspace*{-\leftmargin}{~\textrm{\small{output:} #1}}}

\newcommand{\algbeg}{%
  \addtolength{\labelsep}{0.15in}
  \addtolength{\itemindent}{0.15in}
  \addtolength{\listparindent}{0.15in}
  \addtolength{\linewidth}{-0.15in}
}
\newcommand{\algend}{%
  \addtolength{\labelsep}{-0.15in}
  \addtolength{\itemindent}{-0.15in}
  \addtolength{\linewidth}{0.15in}
}

\newcommand{\comment}[1]{\hfill \raisebox{0pt}{$\ldots$}\hspace{2pt}{\em  #1}}

\makeatletter
\newcommand{\mymathfnnolimits}[1]{\ensuremath{{\mathop {\operator@font\sf #1}\nolimits }}}
\newcommand{\mymathfn}[1]{\ensuremath{\mathop {\operator@font\sf #1}}}
\makeatother
\newcommand{\eps}{\varepsilon}
\newcommand{\tran}{^{\!\scriptscriptstyle \sf T}}

\newcommand{\Vars}{\mymathfnnolimits{{vars}}}
\newcommand{\vars}{\Vars}
\newcommand{\Cons}{\mymathfnnolimits{{cons}}}
\newcommand{\cons}{\Cons}

\newcommand{\calA}{\mathcal{A}}

\newcommand{\calC}{\mathcal{C}}
\newcommand{\calD}{\mathcal{D}}

\newcommand{\calH}{\mathcal{H}}

\newcommand{\calK}{\mathcal{K}}

\newcommand{\calR}{\mathcal{R}}
\newcommand{\calS}{\mathcal{S}}
\newcommand{\calT}{\mathcal{T}}

\newcommand{\R}{{\sf R\hspace*{-1.5ex}\rule{0.15ex}{1.6ex}\hspace*{1.2ex}}}
\renewcommand{\R}{{\sf I\!R}}
\newcommand{\Rp}{\R_{\scriptscriptstyle +}}

\newcommand{\Z}{{\sf \lefteqn{\sf Z}\,Z}}
\newcommand{\Zp}{\Z_{\scriptscriptstyle +}}

\newcommand{\mathspread}{}

\newcommand{\Spread}{\renewcommand{\mathspread}{\,}}
\newcommand{\SPREAD}{\renewcommand{\mathspread}{\,\,\,}}
\newcommand{\dospread}[1]{{\mathspread{#1}\mathspread}}
\Spread

\newcommand{\Ge}{\dospread{\ge}}

\newcommand{\algsep}{
\vspace*{-3.5pt}\hspace*{-1.62em}\rule{1.039\linewidth}{.2pt}}

\newcommand{\stepsize}{\mymathfnnolimits{\sf stepsize}}

\newcommand{\step}{\mymathfnnolimits{\sf Step}}
\newcommand{\pack}{\mymathfnnolimits{\sf IncreasePackingVar}}
\newcommand{\packstar}{\mymathfnnolimits{\sf IncreaseStar}}
\newcommand{\packcomponent}{\mymathfnnolimits{\sf IncreaseComponent}}

\newcommand{\heads}{\mymathfnnolimits{\sf heads}}
\newcommand{\tails}{\mymathfnnolimits{\sf tails}}
\newcommand{\ratio}{\ensuremath{\delta}\xspace}

\newcommand{\prob}[1]{{\sc{#1}}\xspace}

\newcommand{\matching}{\prob{Max Weighted Matching}}
\newcommand{\wvc}{\prob{Weighted Vertex Cover}}
\newcommand{\bmatching}{\prob{Max Weighted} c-\prob{Matching}}
\newcommand{\packing}{\prob{Fractional Packing}}

\newcommand{\distance}{\mymathfnnolimits{\sf distance}}

\renewcommand{\paragraph}[1]{\medskip\par\noindent{\bf #1}~}

\newcommand{\Cite}[1]{{\,\scriptsize\cite{#1}}}

\begin{document}
\title{Distributed Algorithms for Covering, Packing and Maximum Weighted Matching}

\author{
Christos Koufogiannakis\and
Neal E.\ Young
}

\institute{C. Koufogiannakis% \at
		% Department of Computer Science and Engineering\\
		% University of California, Riverside\\
% 		\email{ckou@cs.ucr.edu}
	\and
	N.~E.~Young\at
		Department of Computer Science and Engineering\\
		University of California, Riverside
}

\date{\tt Koufogiannakis, C., Young, N.E. Distributed algorithms for covering, packing and maximum weighted matching. Distributed Computing 24, 45--63 (2011). {https://doi.org/10.1007/s00446-011-0127-7}}

\maketitle

\begin{abstract}
\begin{samepage}
This paper gives poly-logarithmic-round, distributed 
$\ratio$\hyp{approximation} algorithms 
for covering problems with submodular cost
and monotone covering constraints
(\prob{Submodular-cost Covering}).  
The approximation ratio $\ratio$ is the maximum number of variables in any constraint.
Special cases include \prob{Covering Mixed Integer Linear Programs (CMIP)},
and \prob{Weighted Vertex Cover} (with $\ratio=2$).
Via duality, the paper also gives poly-logarithmic-round, distributed $\ratio$\hyp{approximation} algorithms
for \prob{Fractional Packing} linear programs
(where $\ratio$ is the maximum number of constraints in which any variable occurs),
and for \bmatching in hypergraphs
(where $\ratio$ is the maximum size of any of the hyperedges;
for graphs $\ratio=2$).
The paper also gives parallel (RNC) 2-approxi-mation algorithms
for \prob{CMIP} with two variables per constraint
and \prob{Weighted Vertex Cover}.
The algorithms are randomized.
All of the approximation ratios exactly match those
of comparable centralized algorithms.\footnote
{Preliminary versions appeared in \cite{Koufogiannakis09DistributedCovering,Koufogiannakis09DistributedPacking}.
Work by the first author was partially supported by the Greek State Scholarship Foundation (IKY).
Work by the second author was partially supported by NSF awards CNS-0626912, CCF-0729071.}
\end{samepage}
\end{abstract}

% \keywords{Distributed covering, distributed vertex cover, distributed packing, distributed maximum weighted matching} 

\section{Background and results}\label{sec:background}
Many distributed systems are composed of components that can only communicate locally,
yet need to achieve a global (system-wide) goal involving many components.
The general possibilities and limitations of such systems
are widely studied \cite{Linian1992Locality,Naor1995What,Peleg2000Distributed,Kuhn04What,Kuhn06The-price,Czygrinow2008Distributed}.
It is of specific interest to see which fundamental combinatorial optimization problems 
admit efficient distributed algorithms 
that achieve approximation guarantees 
that are as good as those of the best centralized algorithms. 
Research in this spirit includes works on
\prob{Set Cover} (\prob{Dominating Set}) \cite{Jia2002DominatingSet,Kuhn2003Constant,Lenzen2008What},
\prob{capacitated dominating set} \cite{Kuhn2007CapacitatedDS},
\prob{Capacitated Vertex Cover} \cite{Grandoni05Primal-dual,Grandoni2008Primal-Dual},
and many other problems.
This paper presents distributed approximation algorithms 
for some fundamental covering and packing problems.

The algorithms use the standard synchronous communication model:
in each round, nodes can exchange a constant number of messages with neighbors
and perform local computation \cite{Peleg2000Distributed}.
There is no restriction on message size or local computation.
The algorithms are {\em efficient} ---
they finish in a number of rounds that is poly-logarithmic in the network size
 \cite{Linian1992Locality}.

\subsection{Covering Problems}\label{sec:covering problems}
Consider optimization problems of the following form:
{\em 
given a non-decreasing, continuous, and submodular\footnote
{Formally, $c(x)+c(y) \ge c(x\wedge y) + c(x\vee y)$
where $\wedge$ and $\vee$ are component-wise minimum and maximum, respectively.}
cost function $c:\Rp^n\rightarrow \Rp$,
and a set of constraints $\calC$ where 
each constraint $S\in\calC$ is closed upwards\footnote
{If $x\in S$ and $y\ge x$, then $y\in S$.}, 
\begin{center}
\noindent
find $x\in\Rp^n$
minimizing $c(x)$
s.t.~$(\forall S\in \calC)~ x\in S$.
\end{center}
}
\noindent
\prob{Submodular-cost Covering}
includes all problems of this form  \cite{Koufogiannakis2009Covering}.
The (approximation) parameter $\ratio$ 
is the maximum number of elements of $x$ on which any constraint $x\in S$ depends. 

In the above formulation, $x$ ranges over $\Rp^n$, 
but, because the constraints can be non-convex,
one can handle arbitrarily restricted variable domains
by incorporating them into the constraints \cite{Koufogiannakis2009Covering}.
For example, \wvc is equivalent to {\em minimize $\sum_v c_v x_v$
subject to
$x\in\Rp^n$ and

\smallskip
\centerline{
$x_u \ge 1$ or $x_w \ge 1$ ~~~$(\forall (u,w)\in E).$
}
}
\smallskip
\noindent
(Given any 2\hyp{approximate} solution $x$ to this formulation --- which allows $x_u\in \Rp$
--- rounding each $x_u$ down to its floor gives a 2\hyp{approximate} integer solution.)

\prob{Submodular-cost Covering} includes the following problems as special cases:
\begin{description}
\item[\prob{CIP},] covering integer programs with variable upper bounds:
{\em given $A\in\Rp^{m\times n}$, $w\in\Rp^m$, $u\in\Rp^n$,

\smallskip
\centerline{
minimize $c\cdot x$ subject to $x\in\Zp^n$, $Ax \ge w$, and $x\le u$.
}}
\smallskip

\item[\prob{CMIP},]
covering mixed integer linear programs: \prob{CIP} with both integer and fractional variables.

\smallskip
\item[\prob{Facility Location},]
\prob{Weighted Set Cover} (that is, \prob{CIP} with $A_{ij}\in\{0,1\},w_i=1$),
\prob{Weighted Vertex Cover} (that is, \prob{Weighted Set Cover} with $\ratio=2$),
and probabilistic (two-stage stochastic) variants of these problems.
\end{description}

In the centralized (non-distributed) setting there are two well-studied classes
of polynomial-time approximation algorithms for covering problems:

\begin{enumerate}
\item[(i)] $O(\log \Delta)$\hyp{approximation} algorithms
where $\Delta$ is the maximum number of constraints in which any variable occurs
(e.g.~\cite{Johnson73SetCover,Lovasz75SetCover,Chvatal79GreedySetCover,Srinivasan99Improved,Srinivasan01NewApproaches,Kolliopoulos05Approximation}), 
and 

\smallskip
\item[(ii)] $O(\ratio)$\hyp{approximation} algorithms
where $\ratio$ is the maximum number of variables in any constraint
(e.g.\ \cite{Bar-Yehuda81A-Linear-Time,Hochbaum82SetCover,Hochbaum83BoundsSetCover,Bar-Yehuda85A-local-ratio,Hall86A-fast,Bertsimas98Rounding,Carr00Strengthening,Pritchard09Approximability,Bansal2010OnkColumn,Koufogiannakis2009Covering}),
including most famously  2\hyp{approximation} algorithms for \prob{Weighted Vertex Cover}.
\end{enumerate}

The algorithms here match those in class (ii).

\paragraph{Related work.}
For \prob{Unweighted Vertex Cover}, 
it is well known that a 2\hyp{approximate} solution can be found by computing
any maximal matching, then taking the cover to contain the endpoints 
of the edges in the matching.
A maximal matching (and hence a $2$\hyp{approximate} vertex cover)
can be computed deterministically in
$O(\log^4 n)$ rounds using the algorithm of Ha\'{n}\'{c}kowiak, Karonski and Panconesi \cite{Hanckowiak2001}
or in $O(\Delta + \log^*n)$ rounds using the algorithm of Panconesi and Rizzi \cite{Panconesi2001},
where $\Delta$ is the maximum vertex degree.
Or, using randomization, a maximal matching  (and hence a $2$\hyp{approximate} vertex cover)
can be computed in an expected $O(\log n)$
rounds via the algorithm of Israeli and Itai \cite{Israeli1986}.

\prob{Maximal Matching} is also in NC
\cite{Chen1995AFast,Luby1995ASimple,Kelsen1994AnOptimalParallel} and
in RNC --- parallel poly-log time with polynomially many randomized processors
\cite{Israeli1986} --- hence so is 2\hyp{approximating} \prob{Unweighted Vertex Cover}.

For \prob{Weighted Vertex Cover}, previous to this work,
no efficient distributed 2\hyp{approximation} algorithm was known.
Similarly, no parallel NC or RNC 
2\hyp{approximation} algorithm was known.
in 1994, Khuller, Vishkin and Young
gave $2 (1+\eps)$\hyp{approximation} algorithms:
a distributed algorithm taking $O(\log n \log 1/\eps)$ rounds
and a parallel algorithm, in NC for any fixed $\eps$
\cite{Khuller94A-Primal-Dual}.\footnote
{Their result extends to $\ratio(1+\eps)$-approximating \prob{Weighted Set Cover}:
a distributed algorithm taking $O(\ratio \log n \log 1/\eps)$ rounds,
and a parallel algorithm (in NC for fixed $\eps$ and $\ratio$).}
Given integer vertex weights,
their result gives distributed $2$\hyp{approximation} 
in $O(\log n \log n \hat C))$ rounds,
where $\hat C$ is the average vertex weight.
The required number of rounds is reduced to
(expected) $O(\log n \hat C)$
by
Grandoni, K\"onemann and Panconesi 
\cite{Grandoni05Distributed,grandoni2008distributed}.

As noted in \cite{Kuhn04What}, 
neither of these algorithms is efficient 
(taking a number of rounds that is polylogarithmic in the number of vertices).

\newcommand{\rand}{\hfill (random)}

\begin{table*}[t]\label{table:covering}
\noindent\centerline{\small
\begin{tabular}[t]{|l|llrr|} \hline
problem & approx. ratio & \# rounds & where & when\\ \hline \hline
\multirow{2}{*}{\prob{Unweighted Vertex Cover}} 
& 2 & $O(\log n)$ \rand& \Cite{Israeli1986} & 1986\\
& 2 & $O(\log^4 n)$  & \Cite{Hanckowiak2001} & 2001\\
\hline
%neal
%changed next line from \log^{-1} n to \log \eps^{-1} 
\multirow{3}{*}{\prob{Weighted Vertex Cover}} 
& $2+\eps$ & $O(\log \eps^{-1} \log n)$ & \Cite{Khuller94A-Primal-Dual} & 1994\\
& $2$ & $O(\log n \log n\hat C)$ & \Cite{Khuller94A-Primal-Dual} & 1994\\
&2 & $O(\log n\hat C)$ \rand &  \Cite{Grandoni05Distributed,grandoni2008distributed} & 2005 \\
 &2 & $O(\log n)$ \rand &  here & \\
\hline
\multirow{1}{*}{\prob{Weighted Set Cover}}
& $O(\log \Delta)$ & $O(\log m)$ \rand & \Cite{Kuhn06The-price} & 2006\\
& $\ratio$ & $O(\log m)$ \rand & here & \\
\hline
\multirow{1}{*}{\prob{CMIP} (with $\ratio=2$)} 
& $2$ & $O(\log m)$ \rand & here & \\
%\hline
%neal
%changed next line from log to log^2
\multirow{1}{*}{CMIP} 
& $\ratio$ & $O(\log^2 m)$ \rand & here & \\
\hline
\end{tabular}
}
%neal
%check table
\caption{Comparison of distributed algorithms for covering problems.  $\ratio$ is the maximum number of variables in any constraint.  $\Delta$ is the maximum number of constraints in which any variable occurs.}
\vspace*{-1.5em}
\end{table*}

Kuhn, Moscibroda and Wattenhofer describe distributed approximation algorithms for 
{\em fractional} covering and packing linear programs \cite{Kuhn06The-price}.
They show an $O(1)$-approximation with high probability (w.h.p.) in $O(\log m)$ rounds
($m$ is the number of covering constraints).
The approximation ratio is greater than 2 for \prob{Fractional Weighted Vertex Cover}.
For (integer) \prob{Weighted Vertex Cover} and \prob{Weighted Set Cover}
(where each $A_{ij} \in \{0,1\}$)
combining their algorithms with randomized rounding gives
$O(\log\Delta)$\hyp{approximate} integer solutions
in $O(\log n)$ rounds,
where $\Delta$ is the maximum number of constraints in which any variable occurs.

\paragraph{Distributed lower bounds.}
The best lower bounds known for \prob{Weighted Vertex Cover} are by 
Kuhn, Moscibroda and Wattenhofer:
to achieve even a {\em poly-logarithmic} approximation ratio
requires in the worst case $\Omega(\sqrt{\log n / \log\log n})$ rounds.
%neal
%check next line -- "constant"?
In graphs of constant degree $\Delta$,
$\Omega(\log \Delta / \log\log\Delta)$ rounds are required 
\cite{Kuhn04What}.

\paragraph{New results for covering problems.}
This paper gives the following results for covering problems.
\begin{itemize}
\item Section \ref{sec:distributed_vc}
describes the first efficient distributed 2\hyp{approximation} algorithm for \prob{Weighted Vertex Cover}.
The algorithm runs in $O(\log n)$ rounds in expectation and with high probability.

\medskip
\item Section~\ref{sec:distributed_two},
generalizing the above result, describes the first efficient distributed 2\hyp{approximation} algorithm
for \prob{CMIP} (covering mixed integer linear programs with variable upper bounds)
restricted to instances where each constraint has at most two variables ($\ratio = 2$).
The algorithm runs in $O(\log m)$ rounds in expectation and with high probability, where $m$ is the number of constraints.

\medskip
\item Section \ref{sec:distributed_general}
gives the first efficient distributed $\ratio$\hyp{approximation} algorithm for
\prob{Submodular-Cost Covering} (generalizing both problems above).
The algorithm runs in $O(\log^2 m)$ rounds in expectation and with high probability, where $m$ is the number of constraints.
\end{itemize}

Previously, even for the simple special case of
$2$\hyp{approximating} \prob{Weighted Vertex Cover},
no efficient distributed $\ratio$\hyp{approximation} algorithm 
was known.

Each of the algorithms presented here
is a distributed implementation of 
a (centralized) $\ratio$\hyp{approximation} algorithm
for \prob{Submodular-cost Covering}
by Koufogiannakis and Young~\cite{Koufogiannakis2009Covering}.
Each section describes how that centralized algorithm specializes
for the problem in question, then describes an efficient distributed implementation.

\subsection{Fractional Packing, Maximum Weighted Matching}\label{sec:packing}
\packing is the following problem:
{\em
given matrix $A \in \Rp^{n \times m}$ and vectors $w\in \Rp^m$ and $c \in \Rp^n$, 
\begin{center}
maximize $w\cdot y$ subject to $y\in \Rp^m$ and $Ay \le c$.
\end{center}
}
\noindent
This is the linear-program dual of \prob{Fractional Covering} problem 
{\em minimize $c\cdot x$ s.t.~$x \in \Rp^{n}$ and $A^{T}x \ge w$.}

For packing, $\ratio$ is the maximum number of packing constraints in which any variable appears, 
that is, $\max_j |\{i~:~A_{ij}\neq 0\}|$.
In the centralized setting, \packing can be solved optimally in polynomial time using linear programming. 
Alternatively, one can use a faster approximation algorithm 
(e.g., \cite{Koufogiannakis2007Beating}).

\bmatching on a graph (or hypergraph) is the variant where each $A_{ij}\in \{0,1\}$ and 
the solution $y$ must take integer values.
(Without loss of generality each vertex capacity $c_j$ is also an integer.)
An instance is defined by a given hypergraph $H(V,E)$ 
and cost vector $c \in \Zp^{|V|}$; a solution is given by a vector $y \in \Zp^{|E|}$
maximizing $\sum_{e \in E} w_e y_e$ and meeting all the vertex capacity constraints $\sum_{e \in E(u)} y_e \le c_u ~(\forall u \in V)$,
where $E(u)$ is the set of edges incident to vertex $u$.
For this problem, $n = |V|$, $m = |E|$ and $\ratio$ is the maximum (hyper)edge degree (for graphs $\ratio =2$).

\bmatching is a cornerstone optimization problem in
graph theory and Computer Science.
As a special case it includes the ordinary \matching problem (\prob{MWM}) (where $c_u =1$ for all $u\in V$).
Restricted to graphs,
the problem belongs to the ``well-solved  class of integer linear programs"
in the sense that it can be solved optimally in polynomial time
in the centralized setting
\cite{Edmonds1965Paths,Edmonds1970Matching,Schwartz1999Implementing};
moreover
the obvious greedy algorithm (repeatedly select the heaviest edge
that not conflicting with already selected edges)
gives a 2\hyp{approximation}\footnote
{Since it is a maximization problem it is also referred to as a 1/2\hyp{approximation}.}
in nearly linear time.
For hypergraphs the problem is NP-hard --- it generalizes
\prob{set packing}, one of Karp's 21 NP-complete problems \cite{Karp1972}.

\begin{table*}[t]\label{table:packing}
\centering 
\noindent\centerline{\small
\begin{tabular}[t]{|@{~}l@{\,}|l@{~~}l@{~}r@{~}r@{~}|} \hline
problem & approx. ratio & \# rounds & where & when\\ \hline\hline
\multirow{11}{*}{\matching in graphs} 
& $O(\Delta)$ & $O(1)$ & \Cite{Uehara2000Parallel} & 2000\\
& 5 & $O(\log^2 n)$ \rand & \Cite{Wattenhofer04DistributedWeighted} & 2004\\
% & 2 & $O(m)$ & \Cite{Hoepman04Simple} & 2004\\
& $O(1) (>2)$ & $O(\log n)$ \rand & \Cite{Kuhn06The-price,Kuhn2003Constant} & 2006\\
& $4+\eps$ & $O(\eps^{-1} \log \eps^{-1} \log n)$ \rand & \Cite{Lotker2007Distributed}& 2007\\
& $2+\eps$ & $O(\log \eps^{-1} \log n)$ \rand & \Cite{Lotker2008Improved}& 2008\\
& $1+\eps$ & $O(\eps^{-4}\log^2 n)$ \rand & \Cite{Lotker2008Improved}& 2008 \\
& $1+\eps$ & $O(\eps^{-2} + \eps^{-1}\log (\eps^{-1}n) \log n)$ \rand & \Cite{Nieberg2008Local}& 2008\\
& 2 & $O(\log^2 n)$ \rand & \Cite{Lotker2008Improved,Nieberg2008Local} ($\eps =1$)& 2008\\
& $6+\eps$ & $O(\eps^{-1}\log^4 n \log(C_{\max}/C_{\min}))$ &\Cite{panconesi2010fast}  & 2010 \\
& 2 & $O(\log n)$ \rand &  here & \\
\hline
\multirow{2}{*}{\bmatching in graphs} 
& $6+\eps$ & $O(\eps^{-3}\log^3 n \log^2 (C_{\max}/C_{\min}))$ \rand &\Cite{panconesi2010fast}  & 2010 \\
& 2 & $O(\log n)$ \rand &  here & \\
\hline
\multirow{2}{*}{\bmatching in hypergraphs} & $O(\ratio)>\ratio$ & $O(\log m)$ \rand& \Cite{Kuhn06The-price,Kuhn2003Constant} & 2006\\
& $\ratio$ & $O(\log^2 m)$ \rand &  here & \\
\hline
\multirow{2}{*}{\prob{Fractional Packing} ($\ratio=2$)} & $O(1) (>2)$ & $O(\log m)$ \rand& \Cite{Kuhn06The-price} & 2006\\
 &2 & $O(\log m)$ \rand &  here & \\
\hline
\multirow{2}{*}{\prob{Fractional Packing} (general $\ratio$)} & $O(1)>12$ & $O(\log m)$ \rand & \Cite{Kuhn06The-price} & 2006\\
 & $\ratio$ & $O(\log^2 m)$ \rand &  here & \\
\hline
\end{tabular}
}
\caption{Comparison of distributed algorithms for \matching and \packing.}
\vspace*{-1.5em}
\end{table*}

\paragraph{Related work for distributed MWM and Fractional Packing.}
Several previous works consider distributed \matching in graphs.
Uehara and Chen present an $O(\Delta)$\hyp{approximation} algorithm running in a constant 
number of rounds \cite{Uehara2000Parallel}, where $\Delta$ is the maximum vertex degree.
Wattenhofer and Wattenhofer improve this result, giving a randomized $5$\hyp{approximation}
algorithm taking $O(\log^2 n)$ rounds \cite{Wattenhofer04DistributedWeighted}.
Hoepman gives a deterministic $2$\hyp{approximation} algorithm taking $O(m)$ rounds \cite{Hoepman04Simple}.
Lotker, Patt-Shamir and Ros\'{e}n give a randomized $(4+\eps)$\hyp{approximation} algorithm running in
 $O(\eps^{-1} \log \eps^{-1} \log n)$ rounds \cite{Lotker2007Distributed}.
Lotker, Patt-Shamir and Pettie improve this result to a randomized $(2+\eps)$\hyp{approximation}
algorithm taking $O(\log \eps^{-1} \log n)$ rounds \cite{Lotker2008Improved}.
Their algorithm uses as a black box any distributed constant-factor  
approximation algorithm for MWM that takes $O(\log n)$ rounds
(e.g., \cite{Lotker2007Distributed}).
Moreover, they mention (without details) that there is a  distributed $(1+\eps)$\hyp{approximation} algorithm
taking $O(\eps^{-4}\log^2 n)$ rounds, based on the parallel
algorithm by Hougardy and Vinkemeier \cite{Hougardy2006Approximating}.
Nieberg gives a $(1+\eps)$\hyp{approximation} algorithm that runs in 
$O(\eps^{-2} + \eps^{-1}\log (\eps^{-1}n) \log n)$ rounds \cite{Nieberg2008Local}.
The latter two results give randomized $2$\hyp{approximation} algorithms
for \matching in $O(\log^2 n)$ rounds.

Independently of this work, 
Panconesi and Sozio \cite{panconesi2010fast} 
give a randomized distributed $(6+\eps)$-approximation algorithm for
\bmatching in graphs,
requiring $O(\frac{\log^3 n}{\eps^3} \log^2 \frac{C_{\max}}{C_{\min}})$ rounds,
provided all edges weights lie in $[C_{\min},C_{\max}]$.
They also give a similar (but deterministic) result for \matching.

Kuhn, Moscibroda and Wattenhofer show efficient distributed approximation 
algorithms for \packing \cite{Kuhn06The-price}.
They first give a deterministic $(1+\eps)$\hyp{approximation} algorithm for \packing with logarithmic 
message size, but the number of rounds depends on the input coefficients. 
For unbounded message size they give an algorithm for \packing  that finds a constant-factor 
approximation ratio (w.h.p.) in $O(\log m)$ rounds.
If an integer solution is desired, then distributed randomized rounding 
(\cite{Kuhn2003Constant}) can be used.
This gives an $O(\ratio)$\hyp{approximation} for \bmatching on (hyper)graphs
(w.h.p.) in $O(\log m)$ rounds, where $\ratio$ is the maximum
hyperedge degree (for graphs $\ratio=2$). 
(The hidden constant factor in the big-O notation of the approximation ratio
can be relative large compared to a small $\ratio$, say $\ratio = 2$.)

\paragraph{Distributed lower bounds.} 
The best lower bounds known for distributed packing and matching are given by
Kuhn, Moscibroda and Wattenhofer:
to achieve even a poly-logarithmic approximation ratio
for \prob{Fractional Maximum Matching}
takes at least $\Omega(\sqrt{\log n / \log\log n})$ rounds.
%neal
%check next line -- "constant"?
At least $\Omega(\log \Delta / \log\log\Delta)$ rounds
are required
in graphs of constant degree $\Delta$
\cite{Kuhn06The-price}.

\paragraph{Other related work.}
For \prob{Unweighted Maximum Matching} in graphs, Israeli and Itai give a randomized distributed 2\hyp{approximation} algorithm running in $O(\log n)$ rounds \cite{Israeli1986}.
Lotker, Patt-Shamir and Pettie improve this result giving a randomized $(1+\eps)$\hyp{approximation} algorithm taking $O(\eps^{-3} \log n)$ rounds \cite{Lotker2008Improved}.
Czygrinow, Ha\'{n}\'{c}kowiak, and Szyma\'{n}ska show a deterministic 3/2\hyp{approximation} algorithm
that takes $O(\log^4 n)$ rounds \cite{Czygrinow2004AFast}.
A $(1+\eps)$\hyp{approximation} for \matching in graphs
is in NC \cite{Hougardy2006Approximating}.

\paragraph{New results for Fractional Packing and MWM.}~

\noindent
This work presents efficient distributed $\ratio$\hyp{approximation} algorithms for 
\prob{Fractional Packing} and \bmatching.
The algorithms are primal-dual extensions of the $\ratio$\hyp{approximation} algorithms
for covering.
\nocite{Koufogiannakis2009Covering,Bar-Yehuda04Local}. 

\begin{itemize}
\item
Section \ref{sec:distributed_pack_two} describes 
a distributed 2\hyp{approximation} algorithm 
for \packing where each variable appears in at most two 
constraints ($\ratio=2$), 
running in $O(\log m)$ rounds in expectation and with high probability.
Here $m$ is the number of packing variables.
This improves the approximation ratio over the previously best known algorithm
\cite{Kuhn06The-price}.

\medskip
\item
Section \ref{sec:distributed_pack_general} describes
a distributed $\ratio$\hyp{approximation} algorithm
for \packing where each variable appears in at most $\ratio$ constraints, running 
in $O(\log^2 m)$ rounds in expectation and with high probability, 
where $m$ is the number of variables. 
For small $\ratio$, this improves over the best previously known constant
factor approximation \cite{Kuhn06The-price},
but the number of rounds is bigger by a logarithmic-factor. 

\medskip
\item
Section \ref{sec:distributed_pack_two} 
gives a distributed 2\hyp{approximation} algorithm
for \bmatching in graphs,
running in $O(\log n)$ rounds 
in expectation and with high probability.
\bmatching generalizes the well studied \matching problem.
For 2\hyp{approximation}, this algorithm is faster by at least a
logarithmic factor than any previous algorithm. 
Specifically, in $O(\log n)$ rounds, the algorithm gives the best known approximation ratio.
The best previously known algorithms compute a $(1+\eps)$\hyp{approximation} in 
$O(\eps^{-4}\log^2 n)$ rounds \cite{Lotker2008Improved}  
or in $O(\eps^{-2} + \eps^{-1}\log (\eps^{-1}n) \log n)$ rounds \cite{Nieberg2008Local}. 
For a 2\hyp{approximation} each of these algorithms needs $O(\log^2 n)$ rounds.

\medskip
\item
Section \ref{sec:distributed_pack_general} 
also gives a distributed $\ratio$\hyp{approximation} algorithm 
for \bmatching in hypergraphs with maximum hyperedge degree $\ratio$,
running in $O(\log^2 m)$ rounds
in expectation and with high probability, where $m$ is the number of hyperedges. 
This result improves over the best previously known $O(\ratio)$\hyp{approximation} ratio by
\cite{Kuhn06The-price}, but it is slower by a logarithmic factor. 
\end{itemize}

\section{Weighted Vertex Cover}
\label{sec:distributed_vc}

\subsection{Sequential Implementation}\label{sec:sequential vc}
First, consider the following sequential 2\hyp{approximation} algorithm for \prob{Weighted Vertex Cover}.\footnote{For \prob{Weighted Vertex Cover}, 
this sequential algorithm by Koufogiannakis and Young~\cite{Koufogiannakis2009Covering}
is equivalent to the classic 2\hyp{approximation} algorithm
by Bar-Yehuda et al.~\cite{Bar-Yehuda81A-Linear-Time}.}

The algorithm starts with $x=\mathbf 0$.
To cover edge $(v,w)$,  it calls $\step(x,(v,w))$, which raises $x_v$ and $x_w$
at rates inversely proportional to their respective costs,
until $x_v$ or $x_w$ reaches 1
(increase $x_v$ by $\beta/c_v$ and $x_w$ by $\beta/c_w$,
where $\beta = \min\{(1-x_v)c_v, (1-x_w)c_w\}$).
When a variable $x_v$ reaches 1, $v$ is added to the cover.
The algorithm does $\step(x,e)$ for not-yet-covered edges $e$,
until all edges are covered.

\begin{figure*}[t]
\begin{alg}
\Ahead{{\bf Distributed 2-approximation algorithm for Weighted Vertex Cover ($G=(V,E)$, $c: V \rightarrow \Rp$)}}
\Alabel{alg:distributed_vc}

\A At each node $v$: initialize $x_v \leftarrow 0$.

\A Until all vertices are finished, perform rounds as follows:
\algbeg
\Along{{At each node $v$:}
if all of $v$'s edges are covered, finish;
else, choose to be a {\em leaf} or {\em root},
each with probability 1/2.}

\Along{{At each leaf node $v$:}
Label each not-yet covered edge $(v,w)$ {\em active} 
if $w$ is a root and $\step(x,(v,w))$ 
(with the current $x$) would add $v$ to the cover.
Choose, among these active edges,
a random {\em star edge} $(v,w)$.}
\smallskip

\A At each root node $w$,
flip a coin, then run the corresponding subroutine below:
\algbeg

\Anonum{\bf heads$(w)$:}
For each star edge $(v,w)$
(in some fixed order) do: if $w$ is not yet in the cover, then do $\step(x,(v,w))$.

\Anonum{\bf tails$(w)$:}
Do $\step(x,(v,w))$ just for the {\em last} edge for which $\heads(w)$ would do $\step(x,(v,w))$.
\algend
\algend

\algsep

\Ahead{$\step(x,(v,w))$:}
\A Let scalar $\beta \leftarrow \min\big((1-x_v)c_v, (1-x_w)c_w\big)$.  
\comment{just enough to ensure $v$ or $w$ is added to the cover below}
\A Set $x_v \leftarrow x_v + \beta/c_v$.
If $x_v = 1$, add $v$ to the cover, covering all of $v$'s edges.
\A Set $x_w \leftarrow x_w + \beta/c_w$.
If $x_w = 1$, add $w$ to the cover, covering all of $w$'s edges.
\end{alg}
\vspace{-2em}
% \caption{Distributed 2\hyp{approximation} algorithm for Weighted Vertex Cover (\ref{alg:distributed_vc}).}
\end{figure*}

\subsection{Distributed and Parallel Implementations}\label{sec:distributed vc}

In each round, 
the {\em distributed} algorithm simultaneously performs $\step(x,e)$
 on a large subset of the not-yet-covered edges,
as follows.
Each vertex randomly chooses to be a {\em leaf} or a {\em root}.
A not-yet-satisfied edge $(v,w)$ is called {\em active}
if $v$ is a leaf, $w$ is a root, and, if $\step(x,(v,w))$
were to be performed, $v$ would enter the cover.
Each leaf $v$ chooses a random active edge $(v,w)$.
The edges chosen by the leaves are called {\em star edges};
they form stars with roots at their centers.

Each root $w$ then flips a coin.
If heads comes up  (with probability 1/2),
$w$ calls $\heads(w)$,
which does $\step(x,(v,w))$ for its star edges $(v,w)$ in any order,
until $w$ enters the cover or all of $w$'s star edges have steps done.
Or, if tails comes up,
$w$ calls $\tails(w)$,
which simulates $\heads(w)$, without actually doing any steps,
to determine the {\em last} edge $(v,w)$ that $\heads(w)$ would do a step for,
and performs step $\step(x,(v,w))$ for just {\em that} edge.
For details see \ref{alg:distributed_vc}.
\medskip

\begin{theorem}\label{thm:distributed_vc}
For 2-approximating \prob{Weighted Vertex Cover}:

\smallskip
\noindent
(a) There is a distributed algorithm 
running in $O(\log n)$ rounds in expectation and with high probability.

\smallskip
\noindent
(b) There is a parallel algorithm in ``Las Vegas'' RNC.
\end{theorem}
\begin{proof}
(a)
The proof starts by showing that, in each round of \ref{alg:distributed_vc}, 
at least a constant fraction (1/224)
of the not-yet-covered edges are covered in expectation.

Any not-yet-covered edge $(v,w)$ is active for the round with probability at least 1/4,
because $\step(x,(v,w))$ would bring at least one of $v$ or $w$ into the cover,
and with probability 1/4 that node is a leaf and the other is a root.
Thus, by a lower-tail bound, with constant probability (1/7)
at least a constant fraction (one eighth) of the remaining edges are active.\footnote
{In expectation, at most 3/4 of the edges are inactive.
By Markov's inequality, the probability that more than 7/8 of the edges are inactive
is at most (3/4)/(7/8) = 6/7.
Thus, with probability at least 1/7, at least 1/8 of the edges are active.}
Assume at least an eighth of the remaining edges are active.
Next, condition on all the choices of leaves and roots (assume these are fixed).  

It is enough to show that, for an arbitrary leaf $v$,
in expectation at least a quarter of $v$'s active edges will be covered.\footnote
{If so, then by linearity of expectation (summing over the leaves),
at least 1/4 of all active edges will be covered.  
Since a 1/8 of the remaining edges are active (with probability 1/7),
this implies that at least a $1/7*8*4 = 1/224$ fraction of the remaining edges 
are covered in expectation.}
To do so, condition on the star edges chosen by the {\em other} leaves.
(Now the only random choices {\em not} conditioned on 
are $v$'s star-edge choice and the coin flips of the roots.)

\begin{figure}
\centering
\includegraphics[width=0.47\textwidth]{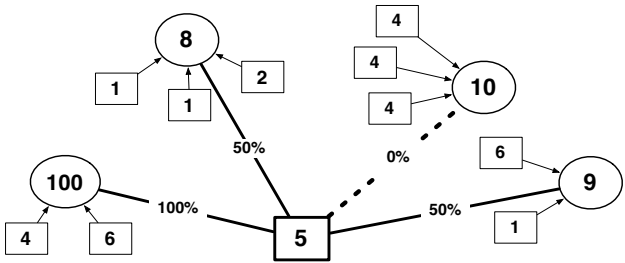}
\caption{Analysis of \ref{alg:distributed_vc}. Each node is labeled with its cost.
Roots are circles; leaves are squares;
star edges from leaves other than $v$ (the cost-5 leaf) are determined as shown.
Each edge $(v,w)$ is labeled with the chance that $v$ would enter the cover
if $v$ were to choose $(v,w)$ for its star edge
(assuming each $x_w=x_v=0$ and each root $w$ considers its star edges counter-clockwise).
}\label{fig:figure}
\vspace{-2em}
\end{figure}

At least one of the following two cases must hold. 

\noindent{{\bf Case 1:}}
{\em At least half of $v$'s active edges $(v,w)$ 
have the following property:
if $v$ were to choose $(v,w)$ as its star edge, and $w$ were to do $\heads(w)$, 
then $\heads(w)$ would {\em not} perform $\step(x,(v,w))$.}
That is, $w$ would enter the cover 
before $\heads(w)$ would consider $(v,w)$
(in Fig.~\ref{fig:figure}, see the cost-10 node).

For such an edge $(v,w)$, on consideration,
$\heads(w)$ will bring $w$ into the cover
{\em whether or not} $v$ chooses $(v,w)$ for its star edge.
So, edge $(v,w)$ will be covered in this round,
regardless of $v$'s choice of star edge, 
as long as $w$ does $\heads(w)$.
Since $w$ does $\heads(w)$ with probability 1/2,
edge $(v,w)$ will be covered with probability 1/2.

Since this is true for at least half of $v$'s active edges,
in expectation, at least a quarter of $v$'s active edges
will be covered during the round.

\noindent{{\bf Case 2:}}
{\em At least half of $v$'s active edges $(v,w)$
have the following property:
if $v$ were to choose $(v,w)$ as its star edge,  and $w$ were to do $\heads(w)$,
then $\heads(w)$ {\em would} perform $\step(x,(v,w))$.}

For such an edge $(v,w)$,
$\heads(w)$ would bring $v$ into the cover
as long as $\step(x,(v,w))$ would not be the {\em last} step performed
by $\heads(w)$
(in Fig.~\ref{fig:figure}, the cost-8 and cost-100 nodes).
Or, if $\step(x,(v,w))$ would be the last step performed by $\heads(w)$,
then $\tails(w)$ would do {\em only} $\step(x,(v,w))$,
which would bring $v$ into the cover
(by the assumption that,
at the start of the round $(v,w)$ is active
so that $\step(x,(v,w))$ would bring $v$ into the cover)
(in Fig.~\ref{fig:figure},  the cost-9 node).
Thus, for such an edge $(v,w)$,
one of $\heads(w)$ or $\tails(w)$
would bring $v$ into the cover.
Recall that $w$ has probability 1/2 of doing $\heads(w)$ 
and probability 1/2 of doing $\tails(w)$.
Thus, if $v$ chooses such an edge,
$v$ enters the cover with at least probability 1/2.

In the case under consideration,
$v$ has probability at least 1/2 of choosing such an edge.
Thus, with probability at least 1/4, $v$ will enter the cover 
and all of $v$'s edges will be deleted.
Thus, in this case also, a quarter of $v$'s edges are covered in expectation
during the round.

Thus, in each round, in expectation at least 1/224 of the remaining edges are covered.
By standard arguments, this implies that the expected
number of rounds is at most about $224\ln(n^2) = 448\ln n$,
and that the number of rounds is $O(\log n)$ with high probability. \footnote
{\label{footnote:karp}
For the high-probability bound,
see e.g.~\cite[Thm.~1.1]{karp1994probabilistic}
(the example following that theorem).
Alternatively, note that, in each round, by a lower-tail bound,
with some constant probability
at least some constant fraction of the remaining edges are covered.
Say that a round is {\em good} if this happens.
By Azuma's inequality, with high probability, 
for some suitably large constants $a < b$,
at least $a\log m$ of the first $b\log m$ rounds are good,
in which case all edges are covered in $b\log m$ rounds.
}

This completes the proof of Thm.~\ref{thm:distributed_vc}, Part (a).

\smallskip
Next is the proof of Part (b).
To obtain the parallel algorithm, implement $\heads(w)$ as follows.
For $w$'s $k$th star edge $e_k$,
let $\beta_k$ be the $\beta$ that $\step(x,e_k)$ would use
if given $x$ at the {\em start} of the round.
If $\heads(w)$ eventually does $\step(x,e_k)$ for edge $e_k$,
the step will increase $x_w$ by $\beta_k/c_w$,
unless $e_k$ is the last edge $\heads(w)$ does a step for.
Thus, the edges for which $\heads(w)$ will do $\step(x,e_k)$
are those for which $x_w + \sum_{j=1}^{k-1} \beta_j/c_w < 1$.
These steps can be identified by a prefix-sum computation,
then all but the last can be done in parallel.
This gives an NC implementation of $\heads(w)$.
The RNC algorithm simulates the distributed algorithm for $O(\log n)$ rounds;
if the simulated algorithm halts, the RNC algorithm returns $x$, and otherwise
it returns ``fail''.  
This completes the proof of Thm.~\ref{thm:distributed_vc}. \qed
\end{proof}

\section{Mixed Integer Programs with Two Variables per Constraint}
\label{sec:distributed_two}
\noindent 
This section generalizes the results of Section~\ref{sec:distributed_vc} to CMIP$_2$ 
(CMIP with at most two non-zero coefficients $A_{ij}$ in each constraint).

\subsection{Sequential Implementation for \prob{CMIP} (any $\ratio$)}
\label{sec:sequential_two}
First, consider the following sequential $\ratio$\hyp{}approximation
algorithm for \prob{CMIP} \cite{Koufogiannakis2009Covering}.
Model the \prob{CMIP} constraints (including the upper bounds and integrality constraints)
by allowing each $x_j$ to range freely in $\Rp$
but replacing each constraint $A_i x \ge b$ 
by the following equivalent constraint $S_i$:
\begin{equation*}
~ \sum_{j\in I} A_{ij}\lfloor \min(x_j,u_{j}) \rfloor 
+ \sum_{j\in \overline I} A_{ij} \min(x_j,u_{j}) \Ge b_i
\end{equation*}
where set $I$ contains the indexes of the integer variables.

The algorithm starts with $x=\mathbf 0$,
then repeatedly does $\step(x,S)$, defined below, for any unsatisfied constraint $S$:

\paragraph{subroutine $\step(x,S)$:}
{\em 
\\ \tab 1. Let $\beta \leftarrow \stepsize(x,S)$.
\\ \tab 2. For each $j$ with $A_{ij}>0$, increase $x_j$ by ${\beta}/{c_j}$.}

\medskip
\noindent
Here $\stepsize(x,S)$ (to be defined shortly) can be smaller than the $\beta$ in \ref{alg:distributed_vc}.  
(Each step might not satisfy its constraint.)

After satisfying all constraints, the algorithm rounds each $x_j$
down to $\lfloor \min(x_j, u_j)\rfloor$ and returns the rounded $x$.
For this to produce a $\ratio$-approximate solution,
it suffices for $\stepsize(x,S)$ to
return a lower bound on the minimum cost of augmenting $x$ to satisfy $S$,
that is, on $\distance_c(x,S)$ $= \min\{c(\hat x) - c(x)|\hat x \in S, \hat x \ge x\}$:

\begin{observation}(\cite{Koufogiannakis2009Covering})\label{obs:stepsize}
If $\stepsize(x,S) \le \distance_c(x,S)$ in each step,
and the algorithm above terminates,
then it returns a $\ratio$-approximate solution.
\end{observation}

\begin{proof}[sketch]
Let $x^*$ be any optimal solution. 
A step increases the cost $c\cdot x$ by $\ratio\beta$,
but decreases the potential $\sum_{v \in V} c_v\max(0, x^*_v - x_v)$
by at least $\beta$.
Details in \cite{Koufogiannakis2009Covering}.
\qed
\end{proof}

Compute $\stepsize(x,S_i)$ as follows.
Consider the relaxations of $S_i$ that can be obtained from $S_i$
by relaxing any subset of the integrality constraints or variable upper bounds.
(That is, replace
$\lfloor \min(x_j,u_j)\rfloor$
by
$\min(x_j,u_j)$
for any subset of the $j$'s in $I$,
and then replace
$\min(x_j,u_j)$
by $x_j$
for any subset of the $j$'s.)
Since there are at most $\ratio$ variables per constraint 
and each can be relaxed (or not) in 4 ways,
there are at most $4^{\ratio}$ such relaxed constraints. 

Define the potential $\Phi(x,S_i)$ of constraint $S_i$
to be the number of these relaxed constraints
not satisfied by the current $x$.
Compute $\stepsize(x,S_i)$ as 
{\em the minimum cost to increase just {\bf one} variable enough to reduce $\Phi(x,S_i)$}.
\begin{observation}\label{obs:hitting}
With this $\stepsize()$,
$\step(x,S_i)$ is done at most $4^{\ratio}$ times before constraint $S_i$ is satisfied.
\end{observation}
(More efficient $\stepsize$ functions are described in \cite{Koufogiannakis2009Covering}.)

This step size satisfies the necessary condition for the algorithm to produce a $\ratio$-approximate solution:

\begin{lemma}\label{lemma:satisfies}
$\stepsize(x,S_i) \le \distance_c(x,S_i)$
\end{lemma}\vspace*{-10pt}
\begin{proof}
Consider a particular relaxed constraint $S'_i$ obtained by relaxing
the upper-bound constraints for all $x_j$ with $x_j < u_j$
and enforcing only a minimal subset $J$ of the floor constraints
(while keeping the constraint unsatisfied).
This gives $S'_i$, which is of the form
\begin{equation*}\SPREAD
~ \sum_{j\in J} A_{ij}\lfloor x_j \rfloor 
+ \sum_{j\in J'} A_{ij} x_j \Ge b_i - \sum_{j\in J''} u_j
\end{equation*}
for some $J$, $J'$, and $J''$.

What is the cheapest way to increase $x$ to satisfy $S'_i$?
Increasing any {\em one} term $\lfloor x_j \rfloor $ for $j\in J$ 
is enough to satisfy $S'_i$
(increasing the left-hand side by $A_{ij}$,
which by the minimality of $J$ must be enough to satisfy the constraint).

Or, if no such term increases,
then  $\sum_{j\in J'} A_{ij} x_j$
must be increased by enough so that increase alone 
is enough to satisfy the constraint.
The cheapest way to do that is to increase
just one variable ($x_j$ for $j\in J'$ maximizing $A_{ij}/c_j$).

In sum, for this $S'_i$,
$\distance(x,S'_i)$ is the minimum cost to increase
just {\em one} variable so as to satisfy $S'_i$.
Thus, by its definition,
$\stepsize(x,S_i) \le \distance(x,S'_i)$.
It follows that

$\stepsize(x,S_i) \le \distance(x,S'_i) \le \distance(x,S_i)$.
\qed
\end{proof}

\vspace*{-10pt}

\noindent \hrulefill
\vspace*{-5pt}

\paragraph{Example.} 
Minimize $x_1 + x_2$
subject to  $0.5 x_1 + 3x_2 \ge 5$,
$ x_2 \le 1$,
and ${x_1,x_2} \in \Zp$.
Each variable has cost 1, so each step will increase each variable equally.
There are eight relaxed constraints:
\begin{align}
0.5 x_1 + 3x_2 \ge 5
\label{eq:relaxed1}\\ 
0.5 x_1 + 3\lfloor x_2 \rfloor \ge 5
\label{eq:relaxed2}\\ 
0.5 x_1 + 3\min\{x_2,1\} \ge 5 
\label{eq:relaxed3}\\ 
0.5 x_1 + 3\lfloor\min\{x_2,1\}\rfloor \ge 5
\label{eq:relaxed4}\\ 
0.5 \lfloor x_1\rfloor+ 3x_2 \ge 5
\label{eq:relaxed5}\\ 
0.5 \lfloor x_1\rfloor + 3\lfloor x_2 \rfloor \ge 5
\label{eq:relaxed6}\\ 
0.5 \lfloor x_1\rfloor + \min\{x_2,1\} \ge 5
\label{eq:relaxed7}\\ 
0.5 \lfloor x_1\rfloor + 3\lfloor\min\{x_2,1\}\rfloor \ge 5
\label{eq:relaxed8}
\end{align}
At the beginning, $x_1=x_2=0$.  
No relaxed constraint is satisfied, so $\Phi(x,S)=8$. 
Then $\stepsize(x,S) = 5/3$
(Constraint (\ref{eq:relaxed1}) or (\ref{eq:relaxed5}) would be satisfied by raising $x_2$ by 5/3).
The first step raises $x_1$ and $x_2$ to $5/3$,
reducing $\Phi(x,S)$ to $6$. 

For the second step, $\stepsize(x,S) = 1/3$
(Constraint (\ref{eq:relaxed2}) or (\ref{eq:relaxed6}) would be satisfied by raising $x_2$ by 1/3).
The step raises both variables by $1/3$ to $2$,
lowering $\Phi(x,S)$ to $4$. 

For the third step, $\stepsize(x,S) = 2$, 
(Constraint (\ref{eq:relaxed3}), (\ref{eq:relaxed4}),
(\ref{eq:relaxed7}), or (\ref{eq:relaxed8}) would be satisfied by raising $x_1$ by 2).
The step raises both variables by 2, to 4,
decreasing $\Phi(x,S)$ to $0$.

All constraints are now satisfied,
and the algorithm returns $x_1 = \lfloor x_1 \rfloor = 4$ and
$x_2 = \lfloor \min\{x_2,1\} \rfloor = 1$.

\noindent \hrulefill

\subsection{Distributed and Parallel Implementations for $\ratio=2$}\label{sec:dist ratio 2}
This section describes a distributed implementation of the above sequential algorithm
for \prob{CMIP$_2$} --- the special case of \prob{CMIP} with $\ratio=2$.
The algorithm (\ref{alg:distributed_two}) generalizes \ref{alg:distributed_vc}.

\begin{figure*}[t]
\begin{alg}
\Ahead{\bf Distributed 2-approximation algorithm for CMIP$_2$
$(c, A, b, u, I)$}
\Alabel{alg:distributed_two}

\Along{At each node $v \in V$: initialize $x_v \leftarrow 0$;
\\if there are unmet constraints $S$ that depend only on $x_v$,
do $\step(x,S)$ for the one maximizing $\stepsize(x,S)$.}

\smallskip
\A Until all vertices are finished, perform rounds as follows:
\algbeg

\Along{At each node $v$: if $v$'s constraints are all met, finish
(round $x_v$ down to $\min(x_v,u_v)$, or $\lfloor \min(x_v,u_v)\rfloor$ if $v\in I$);
\\otherwise, choose to be a {\em leaf} or a {\em root}, each with probability 1/2.}

\smallskip
\Along{Each leaf $v$ does:
for each unmet constraint $S$ 
that can be hit by $x_v$ (Defn.~\ref{defn:hit}),
label $S$ {\em active} if $S$'s other variable is $x_w$ for a root $w$;
choose, among these active constraints,
a random one to be $x_v$'s {\em star constraint} (rooted at $w$).}

\smallskip
\Along{Each root $w$ does either $\heads(w)$ or $\tails(w)$ below, each with probability 1/2.}
\algend

\algsep

\Ahead{\heads$(w)$:}

\Along{For each star constraint $S$ rooted at $w$,
let $t_S$ be the minimum threshold such that increasing $x_w$ to $t_S$
would either hit $S$ (i.e., decrease $\Phi(x,S)$)
or make it so $S$'s leaf variable $x_v$ could no longer hit $S$ (and $x_w$ could).
\\If there is no such value, then take $t_S=\infty$.}

\smallskip
\Along{For each star constraint $S$ rooted at $w$, in order of decreasing $t_S$, do the following:
\\\tab If $x_w< t_S$ then do $\step(x,S)$ (hitting $S$);
otherwise, stop the loop and do the following:
\\Among the star constraints rooted at $w$ that have not yet been hit this round, let $S_r$ (the ``runt'') 
be one maximizing $\stepsize(x,S_r)$.
Do $\step(x,S_r)$ (hitting $S_r$ and all not-yet-hit star constraints rooted at $w$).}

\medskip

\algsep

\Ahead{\tails$(w)$:}
\A Determine which constraint $S_r$ would be the runt in $\heads(w)$. Do $\step(x,S_r)$.
\end{alg}
\vspace{-2em}
% \caption{Distributed 2\hyp{approximation} algorithm for CMIP$_2$ (\ref{alg:distributed_two}).}
\end{figure*}

We assume the network in which the distributed computation takes place
has a node $v$ for every variable $x_v$,
with an edge $(v,w)$ for each constraint $S$ that depends on variables $x_v$ and $x_w$.
(The computation can easily be simulated on, say, a network with
vertices for constraints and edges for variables, or a bipartite 
network with vertices for constraints and variables.)

In \ref{alg:distributed_vc}, a constant fraction of the edges are likely to be covered each round
because a step done for one edge can cover not just that edge, 
but many others also.
The approach here is similar.
Recall the definition of $\Phi(x,S)$ in the definition of $\stepsize()$.
The goal is that the total potential of all constraints,
$\Phi(x) = \sum_S \Phi(x,S)$, should decrease by a constant fraction in each round.

\begin{definition}\label{defn:hit}\em
Say that a constraint $S$ is {\em hit} during the round
when its potential $\Phi(x,S)$ decreases as the result of some step.

By the definition of $\stepsize()$,
for any $x$ and any constraint $S$
there is at least one variable $x_v$ such that raising {\bf just $x_v$} 
to $x_v+\stepsize(x,S)/c_v$ would be enough to hit $S$.
Say such a variable $x_v$ {\em can hit} $S$ (given the current $x$).
\end{definition}
The goal is that a constant fraction of the unmet constraints should be hit in each round.

Note that the definition implies, for example, that, among constraints
that can be hit by a given variable $x_v$,
doing a single step for the constraint $S$ maximizing $\stepsize(x,S)$ 
will hit all such constraints.
Likewise, doing a single step for a random such constraint
will hit in expectation at least half of them
(those with $\stepsize(x,S') \le \stepsize(x,S)$).

\smallskip
In each round of the algorithm,
each node randomly chooses to be a {\em leaf} or a {\em root}.
Each (two-variable) constraint is {\em active} if one of its variables
$x_v$ is a leaf and the other, say $x_w$, is a root, 
and the leaf $x_v$ can hit the constraint at the start of the round.
(Each unmet constraint is active with probability at least 1/4.)
Each leaf $v$ chooses one of its active constraints at random
to be a {\em star constraint}.
Then each root $w$ does (randomly) either $\heads(w)$ or $\tails(w)$,
where $\heads(w)$ does steps for the star constraints rooted at $w$
in a particular order;
and $\tails(w)$ does just one step for the last star constraint
that $\heads(w)$ would have done a step for (called $w$'s ``runt'').

As $\heads(w)$ does steps for the star constraints rooted at $w$, $x_w$ increases.
As $x_w$ increases, the status of a star constraint $S$ rooted at $w$ can change:
it can be hit by the increase in $x_w$
or it can cease to be hittable by $x_v$
(and instead become hittable by $x_w$).
For each constraint $S$, define threshold $t_S$ 
to be the minimum value of $x_w$ at which $S$'s would have such a status change.
Then $\heads(w)$ does steps in order of decreasing $t_S$
until it reaches a constraint $S$ with $x_w \ge t_S$.
At that point, each of $w$'s not-yet-hit star constraints $S$ has $t_S \le x_w$,
and can still be hit by $x_w$.
(As $x_w$ increases, once $S$ changes status,  $S$ will be hittable by $x_w$ 
at least until $S$ is hit.)
Then $\heads(w)$ does $\step(x,S_r)$ for the ``runt'' constraint $S_r$ ---
the one, among $w$'s not-yet-hit star constraints, maximizing $\stepsize(x,S_r)$.
This step hits all of $w$'s not-yet-hit star constraints.
See \ref{alg:distributed_two} for details.

\begin{theorem}\label{thm:distributed_two}
For 2\hyp{approximating} \prob{covering mixed integer linear programs} with at most two variables per constraint (CMIP$_2$):

\smallskip
\noindent
(a) there is a distributed algorithm 
running in $O(\log |\calC|)$ rounds in expectation and with high probability,
where $|\calC|$ is the number of constraints.

\smallskip
\noindent
(b) there is a parallel algorithm in ``Las Vegas'' RNC.
\end{theorem}

The next lemma is useful in the proof of the theorem to follow.
\begin{lemma}
The total potential $\sum_{S_i} \Phi(x,S_i)$ 
decreases by a constant factor in expectation
with each round.
\end{lemma}\vspace*{-10pt}
\begin{proof}
Any unmet constraint is active with probability at least 1/4,
so with constant probability the potential of the active edges
is a constant fraction of the total potential.  Assume this happens.
Consider an arbitrary leaf $v$.
It is enough to show that in expectation a constant fraction
of $v$'s active constraints are hit (have their potentials decrease) during the round.
To do so, condition on any set of choices of star constraints by the {\em other} leaves,
so the only random choices left to be made are $v$'s star-constraint choice
and the coin flips of the roots.
Then (at least) one of the following three cases must hold:

\paragraph{Case 1.}
{\em A constant fraction of $v$'s active constraints $S$ have the following property:
if $v$ were to choose $S$ as its star constraint, 
and the root $w$ of $S$ were to do $\heads(w)$, 
then $\heads(w)$ would {\em not} do $\step(x,S)$.}

Although $\heads(w)$ would not do $\step(x,S)$ for such an $S$,
it nonetheless would hit $S$:
just before $\heads(w)$ does $\step(x,S_r)$,
then $x_w \ge t_S$,
so either $S$ has already been hit (by the increases in $x_w$)
or will be hit by $\step(x,S_r)$
(because $x_w$ can hit $S$
and, by the choice of $S_r$, 
 $\step(x,S_r)$
increases $x_w$ by $\stepsize(x,S_r)/c_w \ge \stepsize(x,S)/c_w$).

On consideration,  for a constraint $S$ with the assumed property,
the steps done by $\heads(w)$ will be the same
even if $v$ chooses some constraint $S'$ with a root other than $w$ as its star constraint.
(Or, if $v$ chooses a constraint $S'\neq S$ that shares root $w$ with $S$,
the steps done by $\heads(w)$ will still raise $x_w$ by as much as
they would have had $v$ chosen $S$ for its star constraint.)
Thus, for such a constraint $S$,
$\heads(w)$ (which $w$ does with probability at least 1/2)
will hit $S$ {\em whether or not} $v$ chooses $S$ as its star constraint.

If a constant fraction of $v$'s active constraints have the assumed property,
then a constant fraction of $v$'s active constraints will be hit with 
probability at least 1/2, so in expectation a constant fraction
of $v$'s active constraints will be hit.

\paragraph{Case 2.}
{\em A constant fraction of $v$'s active constraints $S$
have the following property:
if $v$ were to choose $S$ as its star constraint,
and the root $w$ of $S$ were to do $\heads(w)$,
then $\heads(w)$ would do $\step(x,S)$ when $x_w < t_S$ ($S$ would not be the runt).}

Let $\calH$ denote the set of such constraints.
For $S\in\calH$ let $h(S)$ be the value to which $\heads(w)$ 
(where $w$ is the root of $S$) would increase $x_v$.
Whether or not $v$ chooses $S$ as its star constraint,
if $x_v$ increases to $h(S)$ in the round
and $w$ does $\heads(w)$,
then $S$ will be hit.

Let $S$ and $S'$ be any two constraints in $\calH$ where $h(S) \ge h(S')$.
Let $w$ and $w'$, respectively, be the root vertices of $S$ and $S'$.
(Note that $w=w'$ is possible.)
If $v$ chooses $S'$ as its star constraint
and $w$ and $w'$ both do $\heads()$, 
then $S$ will be hit
(because $x_v$ increases to at least $h(S') \ge h(S)$
and $\heads(w)$ still increases $x_w$ at least to the value
it would have had just before $\heads(w)$ would have done $\step(x,S)$,
if $v$ {\em had} chosen $S$ as its star constraint).

Since (in the case under consideration) a constant fraction of $v$'s active constraints are in $\calH$,
with constant probability $v$ chooses some constraint $S'\in\calH$ as its star constraint
and the root $w'$ of $S'$ does $\heads(w')$.
Condition on this happening.
Then the chosen constraint $S'$ is uniformly random in $\calH$,
so, in expectation, a constant fraction of the constraints $S$ in $\calH$ are hit
(because $h(S) \le h(S')$ and the root $w$ of $S$ also does $\heads(w)$).

\paragraph{Case 3.}
{\em A constant fraction of $v$'s active constraints $S$
have the following property:
if $v$ were to choose $S$ as its star constraint,
and the root $w$ of $S$ were to do $\tails(w)$,
then $\tails(w)$ would do $\step(x,S)$ ($S$ would be the runt).}

Let $\calT$ denote the set of such constraints.
For $S\in\calT$ let $t(S)$ be the value to which $\tails(w)$ 
(where $w$ is the root of $S$) would increase $x_v$.
Whether or not $v$ chooses $S$ as its star constraint,
if $x_v$ increases to $t(S)$ in the round
then $S$ will be hit
(whether or not $w$ does $\tails(w)$).

Let $S$ and $S'$ be any two constraints in $\calT$ where $t(S') \ge t(S)$.
Let $w$ and $w'$, respectively, be the root vertices of $S$ and $S'$.
(Again $w=w'$ is possible.)
If $v$ chooses $S'$ as its star constraint
and $w'$ does $\tails(w')$, 
then (because $x_v$ increases to at least $t(S') \ge t(S)$)
$S$ will be hit.

Since (in the case under consideration) a constant fraction of $v$'s active constraints are in $\calT$,
with constant probability $v$ chooses some constraint $S'\in\calT$ as its star constraint
and the root $w'$ of $S'$ does $\tails(w')$.
Condition on this happening.
Then the chosen constraint $S'$ is uniformly random in $\calT$,
so, in expectation, a constant fraction of the constraints $S$ in $\calT$ are hit
(because $t(S) \le t(S')$).
This proves the lemma. \qed
\end{proof}

\begin{proof}[Thm.~\ref{thm:distributed_two}, part (a)]
The lemma implies that the potential decreases in expectation
by a constant factor each round.
As the potential is initially $O(|\calC|)$ and non-increasing,
standard arguments
(see Footnote~\ref{footnote:karp}) imply
that the number of rounds before the potential is
less than 1 (and so $x$ must be feasible) is $O(\log |\calC|)$
in expectation and with high probability.

This completes the proof of Thm.~\ref{thm:distributed_two}, part (a).\qed
\end{proof}

\subsection*{Parallel (RNC) implementation}
\begin{proof}[Thm.~\ref{thm:distributed_two}, part (b)]
To adapt the proof of (a) to prove part (b), 
the only difficulty is implementing step (2) of $\heads(w)$ in NC.
This can be done using the following observation.
When $\heads(w)$ does $\step(x,S_k)$ for its $k$th star constraint (except the runt),
the effect on $x_w$ is the same as setting
$x_w \leftarrow f_k(x_w)$
for a linear function $f_k$
that can be determined at the start of the round.
By a prefix-sum-like computation,
compute, in NC, for all $i$'s, the functional composition
$F_k = f_k \circ f_{k-1}\circ \cdots \circ f_1$.
Let $x^0_w$ be $x_w$ at the start of the round.
Simulate the steps for all constraints $S_k$ in parallel
by computing $x^k_w = F_k(x^0_w)$,
then, for each $k$ with $x^{k-1}_w < t_{S_k}$,
set the variable $x_v$ of $S_k$'s leaf  $v$
by simulating $\step(x,S_k)$ 
with $x_w = x^{k-1}_w$.
Set $x_w$ to $x^k_w$ for the largest $k$ with $x^{k-1}_w < t_{S_k}$.
Finally, determine the runt $S$ and do $\step(x,S)$.
This completes the description of the NC simulation of $\heads(w)$.

The RNC algorithm will simulate some $c \log |\calC|$ rounds 
of the distributed algorithm, where $c$ is chosen so the
probability of termination is at least 1/2.
If the distributed algorithm terminates in that many rounds,
the RNC algorithm will return the computed $x$.
Otherwise the RNC algorithm will return ``fail''.

This concludes the proof of Thm.~\ref{thm:distributed_two}. \qed
\end{proof}

\section{Submodular-cost Covering}
\label{sec:distributed_general}

\noindent
This section describes an efficient distributed algorithm 
for \prob{Submodular-cost Covering}.
Given a cost function $c$ and a collection of constraints $\calC$, 
the problem is to find $x\in\Rp^n$ to
\smallskip\\\centerline{\em
minimize $c(x)$,
subject to $(\forall S\in \calC)~ x\in S$.
}\smallskip\par\noindent
The cost function $c:\Rp^n\rightarrow \Rp$
is non-decreasing, continuous, and submodular.
Each constraint $S\in\calC$ is closed upwards.

\subsection{Sequential Implementation}
\label{sec:sequential_general}

Here is a brief description of the centralized $\ratio$\hyp{approximation} algorithm
for this problem (for a complete description see \cite{Koufogiannakis2009Covering}).

The algorithm starts with $x=\mathbf 0$,
then repeatedly does the following $\step(x,S)$ for any not-yet-satisfied constraint $S$
(below $\vars(S)$ denotes the variables in $x$ that constraint $S$ depends on):

\paragraph{subroutine $\step(x,S)$:}
\smallskip
{\em 
\\ 1. Let $\beta \leftarrow \stepsize(x,S)$.
\\ 2. For $j\in\vars(S)$,
let $x'_j\in\Rp\cup\{\infty\}$ be maximal s.t.
\\\tab raising $x_j$ to $x'_j$ would raise $c(x)$ by at most $\beta$.
\\ 3. For $j\in\vars(S)$, let $x_j\leftarrow  x'_j$.
}

\medskip
\noindent
Let $\distance_c(x,S)$
denote {\em  the minimum cost of augmenting $x$ to satisfy $S$},
$\min\{ c(\hat x)-c(x) : \hat x\in S,\hat x\ge x\}$.

\begin{observation}(\cite{Koufogiannakis2009Covering})\label{obs:general_stepsize}
If $\stepsize(x,S) \le \distance_c(x,S)$ in each step,
and the algorithm terminates,
then it returns a $\ratio$-approximate solution.
\end{observation}

\noindent
\begin{proof}[sketch, for linear $c$]
Each step starts with $x \not\in S$.
Since the optimal solution $x^*$ is in $S$ and $S$ is closed upwards,
there must be {\em at least one} $k\in\vars(S)$ such that $x_{k} < x^*_{k}$.
Thus, while the algorithm increases the cost of $x$ by at most $\ratio\beta$,
it decreases the potential $\sum_j c_j \max(0, x^*_j - x_j)$ by at least $\beta$.
Full proof in \cite{Koufogiannakis2009Covering}.
\qed
\end{proof}

\paragraph{The function $\stepsize(x,S)$.}
One generic way to define $\stepsize(x,S)$ 
is as the minimum $\beta$ such that $\step(x,S)$ 
will satisfy $S$ in one step. 
This choice satisfies the requirement 
$\stepsize(x,S) \le \distance_c(x,S)$
in Observation~\ref{obs:general_stepsize}
\cite{Koufogiannakis2009Covering}.
But the sequential algorithm above,
and the distributed algorithm described next,
will work correctly with any $\stepsize()$
satisfying the condition in Observation~\ref{obs:general_stepsize}.
For an easy-to-compute and efficient 
$\stepsize()$ function for \prob{CMIP}, 
see \cite{Koufogiannakis2009Covering}.

\subsection{Distributed Implementation}\label{sec:dist general}
Say that the cost function $c(x)$ is {\em locally computable} if,
given any assignment to $x$ and any constraint $S$,
the increase in $c(x)$ due to  $\step(x,S)$ raising the variables in $S$
can be determined knowing
only the values of those variables ($\{x_j~|~j\in\vars(S)\}$).
Any linear or separable cost function is locally computable.

We assume the distributed network
has a node for each constraint $S\in\calC$, with edges from $S$ to
each node whose constraint $S'$ shares variables with $S$
($\vars(S)\cap\vars(S') \neq\emptyset$).
(The computation can easily be simulated on
a network with nodes for variables 
or nodes for variables and constraints.)
We assume unbounded message size.

\begin{theorem}\label{thm:distributed_general}
For $\ratio$\hyp{approximating} any \prob{Submodular-cost Covering} problem
with a locally computable cost function,
there is a randomized distributed algorithm
taking $O(\log^2 |\calC|)$ communication rounds 
in expectation and with high probability,
where $|\calC|$ is the number of constraints. 
\end{theorem}

\begin{proof}
To start each phase, the algorithm 
finds large independent subsets of constraints
by running one phase of Linial and Saks' (LS) decomposition algorithm \cite
{Linial93Low-diameter},~\footnote
{The LS decomposition was also used
in approximation algorithms for fractional packing and covering 
by Kuhn et al.~\cite{Kuhn06The-price}.} below,
with any $k$ such that $k \in \Theta(\log|\calC|)$
(in case the nodes do not know such a value see the comment at the end of this
subsection).
A phase of the LS algorithm, for a given $k$, takes $O(k)$ rounds
and produces a random subset $\calR\subseteq \calC$ of the constraints (nodes),
and for each constraint $S\in\calR$ a ``leader'' node $\ell(S)\in \calS$,
with the following properties:
\begin{itemize}
\item
Each constraint in $\calR$ is within distance $k$ of its leader:

$(\forall S\in \calR)~ d(S, \ell(S)) \le k$.
\item
Edges do not cross components:

$(\forall S, S' \in \calR)~ \ell(S) \neq \ell(S') \rightarrow \vars(S) \cap \vars(S') = \emptyset$.
\item
Each constraint has a chance to be in $\calR$:

$(\forall S\in \calC)~ \Pr[S \in \calR] \ge 1/c|\calC|^{1/k}$
for some $c>1$.
\end{itemize}

Next, each constraint $S\in \calR$ sends its information
(the constraint and its variables' values) to its leader $\ell(S)$.
This takes $O(k)$ rounds because $\ell(S)$
is at distance $O(k)$ from $S$.
Each leader then constructs (locally) the subproblem
induced by the constraints that contacted it and the variables
of those constraints, with their current values.
Using this local copy, the leader repeatedly does $\step(x,S)$
for any not-yet-met constraint $S$ that contacted it,
until all constraints that contacted it are satisfied.

\begin{figure}[t]
\begin{alg1}
\Ahead{\bf Distributed algorithm for problems with arbitrary covering constraints and 
submodular cost}
\Alabel{alg:distributed_general}

\A Initialize $x \leftarrow 0$.
\Along{Compute the Linial/Saks decomposition of the constraint graph $G$.
Denote it $B_1,B_2,\ldots,B_{O(\log |\calC|)}$.}
\A For $b=1,2,\ldots, O(\log |\calC|)$, do:

\algbeg
\A  Within each connected component $\calK$ of block $B_b$:
\algbeg
\A Gather all constraints in $\calK$ at the leader $v_\calK$.
\A{At $v_\calK$, do $\step(x,S)$ until $x\in S$ for all $S\in\calK$.}
\A{Broadcast variables' values to all constraints in $\calK$.}
\algend
\algend
\vspace{-1em}
\end{alg1}
% \caption{Distributed $\ratio$\hyp{approximation} algorithm for problems with arbitrary
% covering constraints and submodular cost (\ref{alg:distributed_general}).}
\vspace{-2em}
\end{figure}

(By the assumption that the cost is locally computable,
the function $\stepsize(x,S)$ and the subroutine $\step(x,S)$
can be implemented knowing only the constraint $S$
and the values of the variables on which $S$ depends.
Thus, the leader can perform $\step(x,S)$ for each
constraint that contacted it in this phase. 
Moreover, distinct leaders' subproblems do not share variables,
so they can proceed simultaneously.)

To end the phase, each leader $\ell$ returns the updated variable information
to the constraints that contacted $\ell$.
Each constraint in $\calR$ is satisfied in the phase
and drops out of the computation 
(it can be removed from the network and from $\calC$;
its variables' values will stabilize once the constraint
and all its neighbors are finished).

\paragraph{Analysis of the number of rounds.}
In each phase (since each constraint is in $\calR$, and thus satisfied, with 
probability $1/c|\calC|^{1/k}$), the number of remaining constraints
decreases by at least a constant factor $1-1/c|\calC|^{1/k} \le 1-1/\Theta(c)$
in expectation.
Thus, the algorithm finishes in $O(c \log|\calC|)$ phases
in expectation and with high probability
$1-1/|\calC|^{O(1)}$.
Since each phase takes $O(k)$ rounds, this proves the theorem.

\paragraph{Comment.}
If the nodes do not know a value $k\in\Theta(\log |\calC|)$,
use a standard doubling trick.
Fix any constant $d>0$.
Start with $x=\mathbf 0$,
then run the algorithm as described above,
except doubling values of $k$ as follows.
For each $k=1,2,4,8,\ldots$,
run $O_d(k)$ phases as described above with that $k$.
(Make the number of phases enough so that,
if $k\ge \ln |\calC|$,
the probability of satisfying all constraints is at least $1-1/|\calC|^d$.)
The total number of rounds
is proportional to the number of rounds in the last group of $O_d(k)$ phases.

To analyze this modification, consider the first $k \ge \log |\calC|$.
By construction, with probability at least $1-1/|\calC|^d$, all constraints
are satisfied after the $O_d(k)$ phases with this $k$.
So the algorithm finishes in $O_d(\log |\calC|)$ phases
with probability at least $1-1/|\calC|^d$.

To analyze the expected number of rounds,
note that the probability of not finishing
in each subsequent group of phases is at most 
$1/|\calC|^d$, while the number of rounds
increases by a factor of four for each increase in $k$,
so the expected number of subsequent rounds is
at most
$O_d(\log|\calC|) \sum_{i=0}^{\infty} 4^i/|\calC|^{di} = O_d(\log|\calC|)$.
\qed
\end{proof}

We remark without proof that the above algorithm can be derandomized at the expense of increasing the number of rounds to super-polylogarithmic (but still sub-linear), using Panconesi and Srinivasan's deterministic variant of the Linial-Saks decomposition \cite{panconesi1996complexity}.

Also, note that although there are parallel (NC) variants of the Linial-Saks decomposition
\cite[Thm.~5]{awerbuch1994low}, this does not yield a parallel algorithm.
For example, if a single variable occurs in all constraints, the underlying network
is complete, so has diameter 1, and thus can be ``decomposed'' into a single cluster.
In the distributed setting, this case can be handled in $O(1)$ rounds
(by doing all the computation at a single node).
But it is not clear how to handle it in the parallel setting.

\subsection{Applications}\label{sec:apps}
\noindent
As mentioned in the introduction, the covering problem considered in this section
generalizes many covering problems.
For all of these, Thm.~\ref{thm:distributed_general} gives a distributed
$\ratio$-approximation algorithm
running in $O(\log^2 |\calC|)$ communication rounds in expectation and with high probability.

\begin{corollary}\label{thm:facility_distributed}
There is a distributed $\ratio$\hyp{approximation} algorithm for 
\prob{Set Cover}, \prob{CMIP} and (non-metric) \prob{Facility Location}
that runs in $O(\log^2 |\calC|)$ communication rounds 
in expectation and with high probability.
\end{corollary}

If the complexity of the computation
(as opposed to just the number of communication rounds) is important,
these problems have appropriate $\stepsize()$ functions
that can be computed efficiently (generally so that each constraint
can be satisfied with overall work nearly linear in the problem size)
\cite{Koufogiannakis2009Covering}. 

\smallskip
Here is a brief description of the problems mentioned above but not previously defined.

(Non-metric) \prob{Facility Location},
for a bipartite graph  $G=(C,F,E)$ of customers $C$ and facilities $F$,
with assignment costs $d$ and opening costs $f$,
asks to find $x\ge 0$ 
such that $\sum_{j\in N(i)} \lfloor x_{ij}\rfloor \ge 1$
for each customer $i$,
while minimizing
the cost to {\em open} facilities
$\sum_{j\in F} f_j \max_{i\in N(j)} x_{ij}$
plus the cost to {\em assign} customers to them $\sum_{ij\in E} d_{ij} x_{ij}$.
The total cost is submodular.
Each customer has at most $\ratio$ accessible facilities.\footnote
{The standard linear-cost formulation is not a covering one.
The standard reduction to set cover increases $\ratio$ exponentially.}

In \prob{probabilistic CMIP},
the constraints are CMIP constraints and
each constraint has a probability $p_S$ of being active.
The stage-one and stage-two costs are specified by
a matrix $w$ and a vector $c$, respectively.
In stage one, the problem instance is revealed.
The algorithm computes, for each constraint $S\in\calC$,
a ``commitment'' vector $y^S\in S$ for that constraint.
The cost for stage one is
$w\cdot y = \sum_{S,j\in \vars(S)} w^S_j y^S_j$.
In stage two, 
each constraint $S$ is (independently) {\em active}
with probability $p_S$.
Let $\calA$ denote the active constraints.
The final solution $x$ is the minimal vector
covering the active-constraint commitments,
i.e. with $x_j = \max\{ y^S_j : S\in\calA, j\in \vars(S)\}$.
The cost for stage two is the random variable $c\cdot x = \sum_j c_j x_j$.
The problem is to choose $y$ to minimize the total expected cost
$C(y) = w\cdot y + E_{\calA}[c\cdot x]$.

\prob{Probabilistic Facility Location} is a 
special case of \prob{probabilistic CMIP}, where
each customer $i$ is also given a probability $p_i$.
In stage one, the algorithm computes $x$
and is charged the assignment cost for $x$.
In stage two, each customer $i$ is {\em active} ($i\in\calA$)
with probability $p_i$, independently.
The algorithm is charged opening costs $f_j$
only for facilities with active customers 
(cost $\sum_j f_j \max_{i\in\cal A} x_{ij}$).
The (submodular) cost $c(x)$ is the total {\em expected} charge. 

%%% Local Variables: 
%%% mode: latex
%%% TeX-master: "distributed_journal"
%%% End: 

\section{Sequential Primal-Dual Algorithm for Fractional Packing and MWM}
\label{sec:covering_packing}

This section gives a sequential $\ratio$-approximation algorithm for \packing
which is the basis of subsequent distributed algorithms for \packing
and \bmatching.
The algorithm is a primal-dual extension of the previous covering algorithms.

\paragraph{Notation.} 
Fix any \prob{Fractional Packing} instance and its \prob{Fractional Covering} dual:
\begin{center}
{\em maximize $w\cdot y$ subject to $y \in \Rp^m$ and $A\tran y \le c$},

{\em minimize $c\cdot x$ subject to $x \in \Rp^{n}$ and $Ax \ge w$.}
\end{center}
Let $C_i$ denote the $i$-th covering constraint ($A_i x \ge w_i$) 
and $P_j$ denote the $j$-th packing constraint ($A^T_j y \le c_j$).
Let $\Vars(S)$ contain the indices of variables in constraint $S$.
Let $\Cons(z)$ contain the indices of constraints in which variable $z$ appears.
Let $N(y_s)$ denote the set of packing variables that share constraints with $y_s$:
$N(y_s) = \vars(\Cons(y_s))$.

\begin{figure}[t]
\begin{alg1}
\Ahead{\bf Sequential algorithm for Fractional Covering}
\Alabel{alg:covering}
\A Initialize $x^0 \leftarrow \mathbf{0}$, $w^0 \leftarrow w$, $t\leftarrow 0$.
\A While $\exists$ an unsatisfied covering constraint $C_s$: 
\algbeg
\A Set $t \leftarrow t+1$.
\comment{do a step for $C_s$}
\label{step:time}
\A Let $\beta_s \leftarrow w^{t-1}_s \cdot \min_{j \in \vars(C_s)}c_{j}/A_{sj}$.
\A For each $j\in\vars(C_s)$:
\algbeg
\A Set $x^t_j \leftarrow x^{t-1}_j + \beta_s /c_j$.
\A For each $i \in \cons(x_j)$ set $w^t_i \leftarrow w^{t-1}_i -A_{ij}\beta_s /c_j$.
\algend
\algend
\A Let $x \leftarrow x^t$.  Return $x$.
\end{alg1}
\vspace{-2em}
% \caption{Sequential $\ratio$\hyp{approximation} algorithm for Fractional Covering (\ref{alg:covering}).}
\end{figure}

\paragraph{Fractional Covering.}
\ref{alg:covering} shows a sequential algorithm for \prob{Fractional Covering}.
Each iteration of the algorithm does one step for an unsatisfied constraint $C_s$:
taking the step size $\beta_s$ to be the minimum such that raising just {\em one} variable $x_j\in\vars(C_s)$
by $\beta_s/c_j$ is enough to make $x$ satisfy $C_s$,
then raising {\em each} variable $x_j$ for $j\in\vars(C_s)$ by $\beta_s/c_j$.

\prob{Fractional Covering} is \prob{CMIP} restricted to instances with no integer variables
and no upper bounds on variables.  
For this special case,
in the sequential algorithm for \prob{CMIP} in Section~\ref{sec:sequential_two},
each constraint has no proper relaxation.  
Hence, that sequential algorithm reduces to \ref{alg:covering}.
In turn, the sequential algorithm for \prob{CMIP} in Section~\ref{sec:sequential_two} 
is a special case of the
sequential algorithm for \prob{Submodular-cost Covering} in Section~\ref{sec:sequential_general}:
\begin{observation}\label{obs:specializes}
\ref{alg:covering} is a special case of both of the following two sequential algorithms:
\begin{itemize}
\item the algorithm for \prob{CMIP} in Section~\ref{sec:sequential_two}, and
\item the algorithm for \prob{Submodular-cost Covering} in Section~\ref{sec:sequential_general}.
\end{itemize}
\end{observation}
The explicit presentation of \ref{alg:covering} eases the analysis. 

\paragraph{Fractional packing.}
Fix a particular execution of \ref{alg:covering} on the \prob{Fractional Covering} instance.
In addition to computing a \prob{Fractional Covering} solution $x$,
the algorithm also (implicitly) defines a \packing solution $y$, as follows:\footnote
{This tail-recursive definition follows local-ratio analyses \cite{Bar-Yehuda05On-the-equivalence}.
The more standard primal-dual approach 
--- setting the packing variable for a covering constraint
when a step for that constraint is done --- doesn't work.
See the appendix for an example.}

Let $T$ be the number of steps performed by \ref{alg:covering}.
For $0\le t \le T$,
after the $t$th step of \ref{alg:covering},
let $x^t$ be the covering solution so far,
and let $w^t = w - A x^t$ be the current slack vector.
Define the {\it residual covering problem} to be 
\begin{center}\em
minimize $c\cdot x$ subject to $x \in \Rp^{n}$ and $A x \ge w^t$;
\end{center}
define the {\it residual packing problem} to be its dual:
\begin{center}\em
maximize $w^t\cdot y$ subject to  $y \in \Rp^m$ and $A\tran y \le c$.
\end{center}
The residual packing constraints are independent of $t$.
%neal
%The algorithm will compute $\ratio$-approximate primal and dual pairs $(x^t,y^{T-t})$ 
Now define a sequence $y^T,y^{T-1}, \ldots, y^1, y^0$, where $y^t$ is
a solution to the residual packing problem after step $t$, inductively, as follows:
Take $y^T = \mathbf 0$.
Then, for each $t=T-1,T-2,\ldots,0$, define $y^t$ from $y^{t+1}$ as follows:
let $C_s$ be the constraint for which \ref{alg:covering} did a step in iteration $t$;
start with $y^t = y^{t+1}$, then raise the single packing variable $y^t_s$
maximally, subject to the packing constraints $A^Ty \le c$.
Finally, define $y$ to be $y^0$.

The next lemma and weak duality prove that this $y$ 
is a $\ratio$-approximation for \prob{Fractional Packing}.
\begin{lemma}\label{thm:2approximation_ratio}
The cost of the covering solution $x$ returned by \ref{alg:covering} 
is at most $\ratio$ times the cost of the packing solution $y$ defined above:
$w\cdot y \ge \frac{1}{\ratio} c\cdot x$.
\end{lemma}
\begin{proof}
When \ref{alg:covering} does a step to satisfy the covering constraint $C_s$
(increasing $x_j$ by $\beta_s /c_j$ for the at most $\ratio$ variables in $\Vars(C_s)$), 
the covering cost $c\cdot x$ increases by at most $\ratio\beta_s$,
so at termination $c\cdot x$ is at most $\sum_{s\in\calD} \ratio \beta_s$.
Since $w\cdot y = w^0 y^0 - w^T y^T = \sum_{t=1}^T (w^{t-1}y^{t-1} - w^ty^t)$,
it is enough to show that $w^{t-1}y^{t-1} - w^ty^t \ge \beta_s$, where $C_s$ is
the covering constraint used in the $t$th step of \ref{alg:covering}.
Show this as follows:

\begin{align}
&~~w^{t-1}y^{t-1} - w^t y^t ~=~ \sum_{i} (w^{t-1}_i y^{t-1}_i - w^t_i y^t_i)
\nonumber\\
&= w^{t-1}_s~ y^{t-1}_s ~+  ~\sum_{i\neq s} (w^{t-1}_i - w^t_i) y^{t-1}_i
\label{eq:second_sequential}
\\
&= y^{t-1}_s\beta_s\max_{j \in \Cons(y_s)} \frac{A_{sj}}{c_j}
~+\sum_{i\ne s; j \in \Cons(y_s)} A_{ij}\frac{\beta_s}{c_j} y^{t-1}_i
\label{eq:fourth_sequential}
\\
&\ge \beta_s ~ \frac{1}{c_j} \sum_{i=1}^m A_{ij}y^{t-1}_i \nonumber\\
&~~\text{(for $j$ s.t. constraint $P_j$ is tight after raising $y_s$)} 
\label{eq:fifth_sequential}
\\
&=\beta_s \nonumber
\end{align}
Eq.\,(\ref{eq:second_sequential}) follows from $y^t_s =0$ and 
$y^{t-1}_i = y^t_i$ $(\forall i \neq s)$.
\\For~(\ref{eq:fourth_sequential}), the definition of $\beta_s$
gives $\displaystyle w^{t-1}_s = \beta_s\max_{j \in \Cons(y_s)} \textstyle \frac{A_{sj}}{c_j}$,
and, by inspecting \ref{alg:covering},
$w^{t-1}_i - w^t_i$
is 
\\$\sum_{j\in\vars(C_s) : i\in \cons(x_j)} A_{ij}\frac{\beta_s}{c_j}
=\sum_{j\in\cons(y_s)} A_{ij}\frac{\beta_s}{c_j}$.
\\Then~(\ref{eq:fifth_sequential}) follows by dropping all but the terms 
for the index $j\in\cons(y_s)$ s.t.~constraint $P_j$ gets tight ($A\tran_j y^{t-1} = c_j$).
The last equality holds because $P_j$ is tight.
\qed
\end{proof}

By the next lemma, the packing solution $y$ has integer entries
as long as the coefficients $A_{ij}$ are 0/1 and the $c_j$'s  are integers:
\begin{lemma}\label{thm:integral_maching}
If $A \in \{0,1\}^{m \times n}$ and $c\in\Zp^n$ then the packing solution $y$ lies in $\Zp^m$.
\end{lemma}
\begin{proof}
Since all non-zero coefficients are 1, the packing constraints are of the form 
$ \sum_{i \in \vars(P_j)} y_i \le c_j ~(\forall i)$.
By (reverse) induction on $t$, each $y^t$ is in $\Zp^m$.
This is true initially, because $y^T=\mathbf 0$.
Assume it is true for $y^{t-1}$.
Then it is true for $y^t$ because, to obtain $y^t$ from $y^{t-1}$, 
one entry $y^t_i$ is raised maximally subject to $A\tran y \le c$.
\qed
\end{proof}

\begin{corollary}
For both \packing and \bmatching, the packing solution $y$
defined above is a $\ratio$-approximation.
\end{corollary}

The covering solution $x$ computed by \ref{alg:covering}
is somewhat independent of the order in which \ref{alg:covering}
considers constraints.
For example, if \ref{alg:covering} does a step for a constraint $C_s$, 
and then (immediately) a step for a different constraint $C_{s'}$ that shares
no variables with $C_s$, then doing those steps in the opposite order
would not change the final covering solution $x$.
Not surprisingly, it also would not change the packing solution $y$ as defined above.
This flexibility is important for the distributed algorithm.
The partial order defined next captures this flexibility precisely:
\begin{definition}\label{def:relation}\em
Let $t_i$ denote the time at which \ref{alg:covering} does a step to cover $C_i$.\footnote
{For now, the time at which a step was performed can be thought as the step number $t$ (line~\ref{step:time} at \ref{alg:covering}). It will be slightly different in the distributed setting.}
Let $t_i = 0$ if no step was performed for $C_i$.
Let $i' \prec i$ denote the predicate
\begin{center}
$\vars(C_{i'}) \cap \vars(C_i) \neq \emptyset$ and $0 < t_{i'} < t_i$.
\end{center}
(That is, the two constraints share a variable
and \ref{alg:covering} did steps for both constraints
--- the first constraint sometime before the second constraint.)

Let $\calD$ be the poset of indices of covering constraints
for which \ref{alg:covering} performed a step,
(partially) ordered by the transitive closure of ``$\prec$". 
\end{definition}

\begin{figure}[t]
\begin{alg1}
\Ahead{\bf Sequential algorithm for Fractional Packing}
\Alabel{alg:packing}
\A Run \ref{alg:covering}; record poset $\calD$ of covering constraint indices.
\A Let $\Pi$ be a total order of $\calD$ respecting the partial order.
\label{step:reverse_order}
\A Initialize $y^T \leftarrow \mathbf{0}$ (where \ref{alg:covering} does $T$ steps).
\A For $t=T-1,T-2,\ldots,0$ do:
\algbeg
\A Set $y^{t} \leftarrow y^{t+1}$.
\A For $s = \Pi_t$, raise $y^{t}_s$ maximally subject to $A\tran y^{t}\le c$:
% neal
% changed from
% $y^{t-1}_s = \max_{j\in\Cons(y_i)}(c_j-\sum_{i=1}^m A_{ij} y^{t-1}_i)$ .\label{line:packing:raise}
% to
\smallskip
\begin{center}
$y^{t}_s \leftarrow \min_{j\in\Cons(y_s)}(c_j-A\tran _j y^t)/A_{sj}$ .\label{line:packing:raise}
\end{center}
\algend
\A Set $y \leftarrow y^0$.  Return $y$.
\end{alg1}
\vspace{-2em}
% \caption{Sequential $\ratio$\hyp{approximation} algorithm for Fractional Packing (\ref{alg:packing}).}
\end{figure}

Sequential \ref{alg:packing} computes the \packing solution $y$ as defined previously,
but considering the variables in an order obtained by reversing {\em any}  total ordering $\Pi$ of $\calD$.\footnote
{The partial order $\calD$ is still defined with respect 
to a particular fixed execution of \ref{alg:covering}.
The goal is to allow \ref{alg:covering} to do steps for the constraints in any order,
thus defining the partial order $\calD$,
and then to allow \ref{alg:packing} to use any total ordering $\Pi$ of $\calD$.}
The next lemma shows that this still leads to the same packing solution $y$.

\begin{lemma}\label{thm:anyorder}
\ref{alg:packing} returns the same solution $y$ 
regardless of the total order $\Pi$ of $\calD$ that it uses
in line~\ref{step:reverse_order}.
\end{lemma}
\begin{proof}
Observe first that, if indices ${s'}$ and $s$ are independent in the poset $\calD$,
then in line~\ref{line:packing:raise} of \ref{alg:packing}, 
the value computed for $y_s$ does not depend on $y_{s'}$ (and vice versa).
This is because, if it did, there would be a $j$ with $A_{sj} \ne 0$ and $A_{s'j} \ne 0$,
so the covering variable $x_j$ would occur in both covering constraints $C_s$ and $C_{s'}$.
In this case, since $C_s$ and $C_{s'}$ share a variable,
$s$ and ${s'}$ would be ordered in $\calD$.

Consider running \ref{alg:packing} using any total order $\Pi$ of $\calD$.
Suppose that, in some iteration $t$, the algorithm raises a variable $y_s$ ($s=\Pi_t$),
and then, in the next iteration, raises a variable $y_{s'}$ ($s'=\Pi_{t-1}$) 
where $y_{s'}$ is independent of $y_s$ in the poset $\calD$.
By the previous observation, the value computed for $y_s$ does not depend on $y_{s'}$, and vice versa.
Thus, {\em transposing} the order in which \ref{alg:packing} considers just these two variables
(so that \ref{alg:packing} raises $y_{s'}$ in iteration $t$ and $y_s$ in the next iteration,
instead of the other way around) would not change the returned vector $y$.
The corresponding total order $\Pi'$ ($\Pi$, but with $s$ and $s'$ transposed)
is also a total ordering of $\calD$.
Thus, for any order $\Pi'$ that can be obtained from $\Pi$
by such a transposition, \ref{alg:packing} gives the same solution $y$.

To complete the proof, observe that {\em any} total ordering $\Pi'$ of $\calD$
can be obtained from $\Pi$  by finitely many such transpositions.
Consider running Bubblesort on input $\Pi$,
ordering elements (for the sort) according to $\Pi'$.
This produces $\Pi'$ in at most $T\choose 2$ transpositions
of adjacent elements.
Bubblesort only transposes two adjacent elements $s',s$,
if $s'$ occurs before $s$ in $\Pi$
but after $s$ in $\Pi'$,
so $s$ and $s'$ must be independent in $\calD$.
Thus, each intermediate order produced along the way
during the sort is a valid total ordering of $\calD$.
\qed
\end{proof}
 
\begin{corollary}
For both \packing and \bmatching, \ref{alg:packing}
is $\ratio$-approximation algorithm.
\end{corollary}

\begin{figure*}[t]
\begin{alg}
\Ahead{\bf Distributed 2-approximation algorithm for Fractional Packing with $\ratio=2$}
\Alabel{alg:distributed_packing_2}

\Ain Graph $G=(V,E)$ representing a fractional packing problem
instance with $\ratio =2$ .
\Aout Feasible $y$, ~2-approximately minimizing $w \cdot y $.

\smallskip
\A Each edge $e_i \in E$ initializes $y_i \leftarrow 0$.
\A Each edge $e_i \in E$ initializes $done_i \leftarrow \text{false}$.
\comment{this indicates if $y_i$ has been set to its final value}

\smallskip
\A Until each edge $e_i$ has set its variable $y_i$ ($done_i = \text{true}$), perform a round:
\algbeg
\A Perform a round of \ref{alg:distributed_two}.
\comment{covering with $\ratio=2$ augmented to compute $(t^R_i$,$t^S_i)$}

\Along{For each node $u_r$ that was a root (in \ref{alg:distributed_two})
at any previous round, consider locally at $u_r$
all stars $\calS^t_r$ that were rooted by $u_r$ at any previous round $t$.
For each star $\calS^t_r$ perform $\packstar(\calS^t_r)$.}
\label{step:reconstruct_stars}
\medskip
\algend

\algsep

\Ahead{$\packstar(\text{star}~\calS^t_r)$:}
\A For each edge $e_i\in\calS^t_r$ in decreasing order of $t^S_i$:
\algbeg
\A If $\pack(e_i) = \text{NOT DONE}$ ~then~BREAK (stop the for loop).
\algend

\algsep

\Ahead{$\pack(\text{edge}~e_i =(u_j,u_r))$:}
\A If $e_i$ or any of its adjacent edges has a non-yet-satisfied covering constraint return NOT DONE.
\smallskip
\A If $t^R_i = 0$ then:
\algbeg
\A Set $y_i \leftarrow 0$ and  $done_i \leftarrow \text{true}$.
% \A Communicate these values to the neighbors and $e$ withdraws from computation.
\A Return DONE.
\algend
\smallskip
\A If $done_{i'} = \text{false}$ for any edge $e_{i'}$ such that
$i \prec {i'}$ then return NOT DONE.
\A Set $y_i \leftarrow \min\big\{(c_j - \sum_{i'} A_{i'j}y_{i'})/A_{ij},~(c_r - \sum_{i'} A_{i'r}y_{i'})/A_{ir} \big\}$~and~
 $done_i \leftarrow \text{true}$.
% \A Communicate these values to the neighbors and $e$ withdraws from computation.
\A Return DONE.
\end{alg}
\vspace{-2em}
% \caption{Distributed 2\hyp{approximation} algorithm for Fractional Packing where each variable appears in at most 2 constraints and Maximum Weighted Matching on graphs (\ref{alg:distributed_packing_2}).}
\end{figure*}

\section{Fractional Packing with $\ratio=2$}
\label{sec:distributed_pack_two}
This section gives a 2\hyp{approximation}  for \packing with $\ratio=2$
(each variable occurs in at most two constraints).

\paragraph{Distributed model.}
\label{sec:model_graphs}
We assume the network in which the distributed computation takes place
has vertices for covering variables (packing constraints) and edges for covering constraints (packing variables).
So, the network has a node $u_j$ for every covering variable $x_j$.
An edge $e_i$ connects vertices $u_j$ and $u_{j'}$ if $x_j$ and $x_{j'}$ belong to the same covering constraint $C_i$, that is, there exists a constraint $A_{ij}x_j + A_{ij'}x_{j'} \ge w_i$ ($\ratio =2$ so there can be at most 2 variables in each covering constraint).

\paragraph{Algorithm.}
Recall \ref{alg:distributed_two},
which $2$-approximates \prob{CMIP$_2$} in $O(\log m)$ rounds.
\prob{Fractional Covering} with $\ratio=2$
is the special case of \prob{CMIP$_2$} with no integer variables and no upper bounds.
Thus, \ref{alg:distributed_two}
also $2$-approximates \prob{Fractional Covering} with $\ratio=2$
in $O(\log m)$ rounds.

Using the results in the previous section,
this section extends \ref{alg:distributed_two}
(as it specializes for \prob{Fractional Covering} with $\ratio=2$)
to compute a solution for the dual problem, \prob{Fractional Packing} with $\ratio=2$,
in $O(\log m)$ rounds.

\ref{alg:distributed_two} proceeds in rounds, and within each round it covers a number of edges. 
Define the time at which a step to cover constraint $C_i$ (edge $e_i$) is 
done as a pair $(t^R_i, t^S_i)$, where $t^R_i$ denotes the round in which the step 
was performed and $t^S_i$ denotes that within the star this step is the $t^S_i$-th one. 
Let $t^R_i = 0$ if no step was performed for $C_i$.
Overloading Definition~\ref{def:relation}, define ``$\prec$" as follows.

\begin{definition}\label{def:relation_graphs}\em
Let ${i'} \prec i$ denote that $\vars(C_{i'}) \cap \vars(C_i) \neq \emptyset$~ 
($i'$ and $i$ are adjacent edges in the network) and
the pair $(t^R_{i'}, t^S_{i'})$ is lexicographically less than 
$(t^R_{i}, t^S_{i})$ (but $t^R_{i'}>0$).
\end{definition}

For any two {\em adjacent} edges $i$ and $i'$,
the pairs $(t^R_i, t^S_i)$ and $(t^R_{i'}, t^S_{i'})$ are adequate to distinguish which edge
had a step to satisfy its covering constraint performed first.  
Adjacent edges can have their covering constraints done in the same round only if 
they belong to the same star (they have a common root), thus they differ in $t^S_i$. 
Otherwise they are done in different rounds, so they differ in $t^R_i$.
Thus the pair $(t^R_i, t^S_i)$ and relation ``$\prec$'' define a partially ordered 
set $\calD$ of all edges for which \ref{alg:distributed_two} did a step.

The extended algorithm, \ref{alg:distributed_packing_2},
runs \ref{alg:distributed_two} to compute a cover $x$, 
recording the poset $\calD$.
Meanwhile, it computes a packing $y$ as follows:
as it discovers $\calD$, it starts raising packing variable as soon as it can
(while emulating \ref{alg:packing}).
Specifically it sets a given $y_i\in\calD$
as soon as
(a) \ref{alg:distributed_two} has done a step for the covering constraint $C_i$,
(b) the current cover $x$ satisfies each adjacent covering constraint $C_{i'}$, and 
(c) for each adjacent $C_{i'}$ for which \ref{alg:distributed_two} did a step {\em after} $C_i$,
the variable $y_{i'}$ has already been set.
When it sets the packing variable $y_i$, 
it does so following \ref{alg:packing}:
it raises $y_i$ maximally subject to the packing constraint $A\tran y \le c$.

Some nodes will be executing the second phase of the algorithm 
(computing the $y_i$'s)
while some other nodes are still executing the first phase
(computing the $x_j$'s).
This is necessary because a given node cannot know when distant nodes are done 
computing $x$.

\begin{theorem}\label{thm:packing_time_graphs}
For \packing where each variable appears in at most two constraints ($\ratio = 2$)
there is a distributed 2-approximation algorithm running
in $O(\log m)$ rounds in expectation and with high probability, where
$m$ is the number of packing variables.
\end{theorem}
\begin{proof}
By Thm.~\ref{thm:distributed_two}, \ref{alg:distributed_two} computes a covering solution $x$ 
in $T=O(\log m)$ rounds in expectation and with high probability.
Then by a straightforward induction on $t$, within $T+t$ rounds, 
for every constraint $C_i$ for which \ref{alg:distributed_two} did a step in round $T-t$,
\ref{alg:distributed_packing_2} will have set the variable $y_i$.
Thus, \ref{alg:distributed_packing_2} computes the entire packing $y$
in $2T=O(\log m)$ rounds with in expectation and with high probability.

As \ref{alg:distributed_two} performs steps for covering constraints,
order the indices of those constraints by the time in which the constraints' steps were done,
breaking ties arbitrarily.  Let $\Pi$ be the resulting order.

As \ref{alg:distributed_packing_2} performs steps for packing variables,
order the indices of those variables by the time in which the variables were raised,
breaking ties arbitrarily.  Let $\Pi'$ be the reverse of this order.

Then the poset $\calD$ is the same as it would be if defined
by executing the sequential algorithm~\ref{alg:covering}
for the \prob{Fractional Covering} problem 
and considering the constraints in the order in which their indices occur in $\Pi$.
Also, the covering solution $x$ is the same as would
be computed by \ref{alg:covering}.

Likewise, $\Pi'$ is a total ordering of $\calD$,
and the packing solution $y$ is the same as would
be computed by \ref{alg:packing} using the order $\Pi'$.

By Lemmas~\ref{thm:2approximation_ratio} and~\ref{thm:anyorder},
and weak duality, the \prob{Fractional Packing} solution $y$ is 2-approximate. 
\iffalse

To finish, we prove that the $y$ can be computed in at most $T$ extra rounds after
the initial $T$ rounds to compute $x$. 
First note that within a star, even though its edges are ordered according
to $t^S_i$ they can all set their packing variables in a single round if none of 
them waits for some adjacent edge packing variable that belongs to a different star.
So in the rest of the proof, consider only the case were edges are waiting for adjacent
edges that belong to different stars. 
Note that $1 \le t^R_i \le T$ for each $y_i\in \calD$.
Then, at round $T$, each $y_i$ with $t^R_i = T$ can be set in this round because
it does not have to wait for any other packing variable to be set.
At the next round, round $T+1$, each $y_i$ with $t^R_i = T-1$ can be set; they 
are dependent only on variables $y_{i'}$ with $t^R_{i'} = T$ which have been already set.
In general, packing variables with $t^R_i = t$ can be set once
all adjacent $y_{i'}$ with $t^R_i \ge t+1$ have been set.
Thus by induction on $t=0,1,\ldots$ 
a constraint $C_i$ for which a step was done at round $T-t$ 
may have to wait until at most round $T+t$ until its packing variable $y_i$ is set.
Therefore, the total number of rounds until solution $y$ is
computed is $2T = O(\log m)$ in expectation and with high probability.
\fi
\qed
\end{proof}

The following corollary is a direct result of Lemma~\ref{thm:integral_maching}
and Thm.~\ref{thm:packing_time_graphs} and the fact that for 
this problem $|E| = O(|V|^2)$.
\begin{corollary}\label{thm:distributed_b-matching}
There is a distributed $2$-approximation algorithm 
for \bmatching on graphs 
running in $O(\log |V|)$ rounds
in expectation and with high probability.
\end{corollary}

\section{Fractional Packing with general $\ratio$}
\label{sec:distributed_pack_general}

\paragraph{Distributed model.}
\label{sec:model_hypergraphs}
Here we assume that the distributed network has a node $v_i$ for
each covering constraint $C_i$ (packing variable $y_i$), with edges from $v_i$ to
each node $v_{i'}$ if  $C_i$ and $C_{i'}$ share a covering variable $x_j$. \footnote
{The computation can easily be simulated on a network with nodes for covering variables or nodes for covering variables and covering constraints.}.
The total number of nodes in the network is $m$.
Note that in this model the role of nodes and edges is reversed as compared to the
model used in Section~\ref{sec:distributed_pack_two}.

\paragraph{Algorithm.}
Fix any \prob{Fractional Packing} instance and its \prob{Fractional Covering} dual:
\begin{center}
{\em maximize $w\cdot y$ subject to $y \in \Rp^m$ and $A\tran y \le c$},

{\em minimize $c\cdot x$ subject to $x \in \Rp^{n}$ and $Ax \ge w$.}
\end{center}

Recall \ref{alg:distributed_general},
which $\ratio$-approximates \prob{Submodular-cost Covering} in $O(\log^2 m)$ rounds.
\prob{Fractional Covering}  (with general $\ratio$) is a special case.
Thus, \ref{alg:distributed_general}
also $\ratio$-approximates \prob{Fractional Covering}
in $O(\log^2 m)$ rounds.

For this special case, make \ref{alg:distributed_general}
use the particular $\stepsize()$ function defined in \ref{alg:covering}
so that the specialization of \ref{alg:distributed_general} defined thusly
is also a distributed implementation of \ref{alg:covering}.

Following the approach in the previous two sections,
this section extends this specialization of \ref{alg:distributed_general} for \prob{Fractional Covering}
to compute a solution for \prob{Fractional Packing},
the dual problem, in $O(\log^2 m)$ rounds.

Similar to the $\ratio=2$ case in the previous section, 
\ref{alg:distributed_general} defines a poset $\calD$ 
of the indices of covering constraints;
the extended algorithm sets raises packing variables maximally,
in a distributed way consistent with some total ordering of $\calD$.
Here, the role of stars is substituted by components and the role of roots by leaders. 
With each step done to satisfy the covering constraints $C_i$, the 
algorithm records $(t^R_i, t^S_i)$, where $t^R_i$ is the round and $t^S_i$ is the within-the-component 
iteration in which the step was performed.
This defines a poset $\calD$ on the indices of covering constraints 
for \ref{alg:distributed_general} performs steps.  

\begin{figure*}[t]
\begin{alg}
\Ahead{\bf Distributed $\ratio$-approximation algorithm for Fractional Packing with general $\ratio$}
\Alabel{alg:distributed_packing_degree}

\Ain Graph $G=(V,E)$ representing a fractional packing problem instance.
\Aout Feasible $y$, ~$\ratio$-approximately minimizing $w \cdot y $.

\smallskip
\A Initialize $y \leftarrow 0$.
\A For each $i=1\dots m$ initialize $done_i \leftarrow \text{false}$.
\comment{this indicates if $y_i$ has been set to its final value}

\A Until each $y_i$ has been set ($done_i = \text{true}$) do:
\algbeg
\Along{Perform a phase of the $\ratio$-approximation algorithm for covering (\ref{alg:distributed_general}), recording $(t^R_i,t^S_i)$.}

\Along{For each node $v_\calK$ that was a leader at any previous phase, consider locally at $v_\calK$
all components that chose $v_\calK$ as a leader at any previous phase.
For each such component $\calK_r$ perform $\packcomponent(\calK_r)$.}
\medskip
\algend

\algsep

\Ahead{$\packcomponent(\text{component}~\calK_r)$:}
\A For each $i\in\calK_r$ in decreasing order of $t^S_i$:
\algbeg
\A If $\pack(i) = \text{NOT DONE}$ ~then~BREAK (stop the for loop).
\algend

\algsep

\Ahead{$\pack(i)$:}
\Along{If $C_i$ or any $C_{i'}$ that shares covering variables with $C_i$ is not yet satisfied return NOT DONE.}
\smallskip
\A If $t^R_i = 0$ then:
\algbeg
\A Set $y_i = 0$ and  $done_i = \text{true}$.
\A Return DONE.
\algend
\smallskip
\A If $done_{i'} = \text{false}$ for any $y_{i'}$ such that $i \prec {i'}$ then return NOT DONE.
\A Set $y_i \leftarrow \min_{j\in \cons(y_i)} \big((c_j - \sum_{i'} A_{i'j}y_{i'})/A_{ij})\big)$~and~
 $done_i \leftarrow \text{true}$.
\A Return DONE.
\end{alg}
\vspace{-2em}
% \caption{Distributed $\ratio$\hyp{approximation} algorithm for Fractional Packing and Maximum Weighted Matching on hypergraphs (\ref{alg:distributed_packing_degree}).}
\end{figure*}

The extended algorithm, 
\ref{alg:distributed_packing_degree},
runs the specialized version of \ref{alg:distributed_general} 
on the \prob{Fractional Covering} instance,
recording the poset $\calD$
(that is, recording $(t^R_i, t^S_i)$ for each covering constraint $C_i$ for which it performs a step).
Meanwhile, as it discovers the partial order $\calD$, it begins computing the
packing solution, raising each packing variable as soon as it can.
Specifically sets a given $y_i\in\calD$ once
(a) a step has been done for the covering constraint $C_i$,
(b) each adjacent covering constraint $C_{i'}$ is satisfied and 
(c) for each adjacent $C_{i'}$ for which a step was done after $C_i$,
the variable $y_{i'}$ has been set.

To implement this, the algorithm considers all components that have been done by leaders
in previous rounds.
For each component, the leader considers the component's packing variables $y_i$ in order of decreasing $t^S_i$.
When considering $y_i$ it checks if each $y_{i'}$ with $i \prec {i'}$ 
is set, and, if so, the algorithm sets $y_i$ 
and continues with the next component's packing variable (in order of decreasing $t^S_i$).
%neal
%Otherwise the algorithm cannot yet decide about the remaining component's packing variables.

\begin{theorem}\label{thm:packing_time_hypergraphs}
For \packing where each variable appears in at most $\ratio$ constraints
there is a distributed $\ratio$-approximation algorithm
running in $O(\log^2 m)$ rounds in expectation and with high probability,
where $m$ is the number of packing variables.
\end{theorem}
\begin{proof}
The proofs of correctness and running time 
are essentially the same as in the proofs of Thm.~\ref{thm:packing_time_graphs},
except that in this case $T=O(\log^2 m)$ in expectation and w.h.p.
(by Thm.~\ref{thm:distributed_general}), so the running time 
is $O(\log^2 m)$ .
\qed
\end{proof}

\vspace{10pt}
The following corollary is a direct result of Lemma~\ref{thm:integral_maching}
and Thm.~\ref{thm:packing_time_hypergraphs}.

\begin{corollary}\label{thm:distributed_b-matching_hypergraphs}
For \bmatching on hypergraphs,
there is a distributed $\ratio$-approximation algorithm running
in $O(\log^2 |E|)$ rounds in expectation and with high probability,
where $\ratio$ is the maximum hyperedge degree and $|E|$ is the number of hyperedges.
\end{corollary}

%\begin{acknowledgements}
\section*{Acknowledgements}
Thanks to two anonymous referees for their helpful comments.
%\end{acknowledgements}

\section*{Appendix}

Generate a $\ratio$-approximate primal-dual pair
for the greedy algorithm for \prob{Fractional Covering} (\ref{alg:covering} in Section~\ref{sec:covering_packing}) in some sense requires a tail-recursive approach.
The following example demonstrates this.
Consider (i)

\noindent
$\min\{x_1+x_2+x_3 : x_1+x_2\ge 1,~x_1 + x_3 \ge 5,~x \ge 0\}$.

\smallskip
\noindent
If the greedy algorithm (\ref{alg:covering}) does the constraints of (i) in either order
(choosing $\beta$ maximally), it gives a solution of cost 10.

The dual is
$\max\{y_{12} + 5 y_{13} : y_{12}+y_{13}\le 1,~y \ge 0\}$.
The only way to generate a dual solution of cost 5
is to set $y_{12}=0$ and $y_{13}=1$.

Now consider the covering problem (ii)

\noindent
$\min\{x_1+x_2+x_3 : x_1+x_2\ge 1,~x_1 + x_3 \ge 0,~x \ge 0\}$
(the right-hand-side of the second constraint is 0).

If the greedy algorithm does the constraints of (ii) in either order
(choosing $\beta$ maximally), it gives a solution of cost 2.

The dual is $\max\{y_{12} : y_{12}+y_{13}\le 1,~y \ge 0\}$.
The only way to generate a dual solution of cost 1
is to set $y_{12}=1$ and $y_{12}=0$
(the opposite of the $y$ for (i)).

Consider running \ref{alg:covering} on each problem,
and giving it the shared constraint $x_1+x_2\ge 1$ first.
This constraint (and the cost) are the same in both problems,
and the algorithm sets $x$ to satisfy the constraint in the same way for both problems.
A standard primal-dual approach would set the dual variable $y_{12}$ at this point,
as a function of the treatment of $x_1$ and $x_2$
(essentially in an {\em online} fashion, constraint by constraint). 
That can't work here:
in order to get a 2-approximate dual solution,
the dual variable $y_{12}$ has to be set differently
in the two instances: to 0 for the first instance, and to 1 for the second instance.

%\bibliographystyle{abbrvnat}
% citation styles above seem broken

% longer citations would mess up the tables,
% so please don't change the next line:
\bibliographystyle{plain}
\bibliography{merged}

\end{document}

%%% Local Variables: 
%%% mode: latex
%%% TeX-master: t
%%% End: 